\documentclass[journal,twocolumn]{IEEEtran}

\usepackage{amsmath, amssymb,bm,amssymb}
\usepackage{mathrsfs}
\usepackage{bbm}
\usepackage{url}
\usepackage{graphicx}
\usepackage{cite}
\usepackage{esvect}
\usepackage{subcaption}
\usepackage{balance}
\usepackage{bm}

\usepackage{multirow}
\usepackage{array}
\usepackage{ragged2e}
\usepackage{enumerate} 
\usepackage{booktabs}
\usepackage{tikz}

\usepackage{float}
\usepackage{setspace}

\usepackage{amsthm}
\usepackage{epstopdf}
\newtheorem{theorem}{Theorem}

\newtheorem{proposition}{Proposition}
\newtheorem{lemma}[theorem]{Lemma}

\usepackage[ruled,linesnumbered]{algorithm2e} 
\usepackage{caption}

\allowdisplaybreaks[1]
\allowdisplaybreaks
\usepackage{comment}
\usepackage{pifont}
\usepackage{color}

\ifCLASSINFOpdf
\else
\fi

\begin{document}
	\title{Reconfigurable Intelligent Surfaces Empowered Green Wireless Networks with User Admission Control
	}
	
	\author{Jinglian He, \emph{Student Member, IEEE}, Yijie Mao, \emph{Member, IEEE}, Yong Zhou, \emph{Member, IEEE}, Ting Wang, \emph{Senior Member, IEEE}, Yuanming Shi, \emph{Senior Member, IEEE}
		
		\thanks{This work was supported in part by the Natural Science Foundation of Shanghai under Grant No. 21ZR1442700, in part by Shanghai Rising-Star Program under Grant No. 22QA1406100, in part by the National Natural Science Foundation of China (NSFC) under grant 62001294, in part by the National Nature Science Foundation of China under Grant 62201347, and in part by Shanghai Sailing Program under Grant 22YF1428400.
			
			J. He is with the School of Information Science and Technology, ShanghaiTech
			University, Shanghai 201210, China, Shanghai Institute of Microsystem
			and Information Technology, Chinese Academy of Sciences, Shanghai
			200050, China, and also with the University of Chinese Academy of Sciences,
			Beijing 100049, China (e-mail: hejl1@shanghaitech.edu.cn).
			
			Y. Mao Y. Zhou and Y. Shi are with the School of Information Science and Technology, ShanghaiTech University, Shanghai 201210, China (e-mail: \{maoyj, zhouyong, shiym\}@shanghaitech.edu.cn).
			
			Ting Wang is with the Software Engineering Institute, Shanghai Key Laboratory of Trustworthy Computing, East China Normal University, Shanghai, China (twang@sei.ecnu.edu.cn).
	}
		\thanks{\emph{Corresponding author: Yijie Mao.}}
	}
	
	\maketitle
	
	\thispagestyle{empty}
	
	\IEEEpeerreviewmaketitle

	\begin{abstract}
		Reconfigurable intelligent surface (RIS) has emerged as a cost-effective and energy-efficient technique for 6G.
		By adjusting the phase shifts of passive reflecting elements, RIS is capable of  suppressing the interference and combining the desired signals constructively at receivers, thereby significantly enhancing the performance of communication system.
		In this paper, we consider a green multi-user multi-antenna cellular network, where multiple RISs are deployed to provide energy-efficient communication service to end users.
		We jointly optimize the phase shifts of RISs, beamforming of the base stations, and the active RIS set with the aim of minimizing the power consumption of the base station (BS) and RISs subject to the quality of service (QoS) constraints of users and the transmit power constraint of the BS. However, the problem is mixed combinatorial and non-convex,  and there is a potential infeasibility issue when the QoS constraints cannot be guaranteed by all users. 
		To deal with the infeasibility issue, we further investigate a user admission control problem to jointly optimize the transmit beamforming, RIS phase shifts, and the admitted user set. 
		A unified alternating optimization (AO) framework is then proposed to solve both the power minimization and user admission control problems.
		Specifically, we first decompose the original non-convex problem into several rank-one constrained optimization subproblems via matrix lifting.
		A difference-of-convex (DC) algorithm is then developed to solve each decomposed subproblem.
		The proposed AO framework efficiently minimizes the power consumption of wireless networks as well as user admission control when the QoS constraints cannot be guaranteed by all users.
		To further address the complexity-sensitive issue for practical implementation, we propose an alternative low-complexity beamforming and RISs phase shifts design algorithm based on zero-forcing (ZF)  to enable the green cellular networks.
		
	\end{abstract}
	
	
	\section{Introduction}
	
	In recent years, reconfigurable intelligent surface (RIS) has emerged as a promising technique for 6G \cite{rajatheva2020white,yuan2021reconfigurable,yan2019passive, yuan2021wc}.
	In particular, RIS is a two-dimensional meta-surface consisting of a large number of passive reflecting elements with reconfigurable phase shifts, each of which can be dynamically adjusted by a software controller \cite{wu2019ris, christos2018meta}.
	Due to the thin films form, RIS can be deployed onto the walls of high-rise buildings with a low deployment cost, which makes it easy to be integrated into wireless communications\cite{basar2019wireless, wang2020channel,he2020coordinated,wang2022tgcn}.
	RIS fundamentally changes the electromagnetic wave propagation and offers a high quality of service (QoS) for users in cellular networks \cite{wu2019towards}.
	By jointly optimizing the beamforming vectors of the base station (BS) and the phase shifts of RIS, the RIS-assisted wireless network is capable of improving the signal propagation conditions\cite{bjornson2020rayleigh}, enhancing the spectral efficiency\cite{guo2020weighted,shu2021tcom}, and reducing the power consumption \cite{huang2019reconfigurable,shi2021cc}.
	To ensure green and sustainable wireless networks, energy efficiency is a key metric in 6G networks \cite{mahapatra2016ee}.
	RIS is regarded as a promising approach to improve the energy efficiency in 6G.
	Thanks to its passive reflecting elements, RIS forwards the incoming signal without employing any power amplifier.
	The power consumption of RIS is therefore much lower than traditional methods such as amplify and forward (AF) relays and backscatter communications \cite{wu2019intelligent}.

	A number of efforts have been devoted to improve the power efficiency of wireless communication based on RIS.
	In \cite{hua2021greenris}, the authors designed a task selection strategy and leveraged a RIS to reduce the overall power consumption in green edge interference networks.
	The authors in \cite{zhou2020hardware} investigated the hardware impairments in RIS-assisted multiple-input single-output (MISO) broadcast channels (BC), and derived a closed-form solution for maximizing the energy efficiency. 
	However, both \cite{hua2021greenris} and \cite{zhou2020hardware} only considered the power consumption of the BS.
	Due to the hardware static power consumption of the RISs introduced by phase shifts adjusting, the power consumption of RIS should also be considered when maximizing the energy efficiency of networks, e.g., considering both the power consumption of the BS and that of all RISs.
	In addition to the limitation of the power consumption model, \cite{hua2021greenris} and \cite{zhou2020hardware} only deployed one RIS  in the considered networks.
	Moving to the power consumption of multiple RISs, the authors in \cite{wang2020active} leveraged multiple RISs to maximize the received power for downlink point-to-point millimeter wave communications.
	The authors in \cite{yang2021distributed} investigated the RIS selection problem by taking the power consumption of RISs into consideration. 
	In \cite{sun2020green}, the authors studied the power minimization problem by jointly optimizing the transmit beamforming at the BS and multi-RIS selection.
	These works dynamically controlled the on-off status of each RIS as well as optimized the RIS phase shifts of the active RISs in order to maximize the network energy efficiency.

	\begin{table*}[t]
		\centering
		\caption{Survey of existing works for green cellular network.}
		\label{table1}
		\resizebox{\linewidth}{!}{
			\begin{tabular}{|c|c|c|c|c|c|c|c|c|}
				\hline
				\multirow{2}{*}{Work} & \multicolumn{2}{c|}{Optimization problem} & \multicolumn{4}{c|}{Optimization variables}& \multirow{2}{*}{Number of RISs}\\
				\cline{2-3}\cline{4-7}
				& Network power minimization & User admission control & BS beamforming  &RIS phase shifts & Active RIS set& Active user set&\\
				\hline
				\cite{hua2021greenris} &  \checkmark& \ding{53}&  \checkmark&\checkmark &\ding{53} &\ding{53} &Single\\
				\hline
				\cite{zhou2020hardware} &  \checkmark& \ding{53}& \checkmark&  \checkmark&\ding{53} & \ding{53}&Single\\
				\hline
				\cite{wang2020active} &  \checkmark& \ding{53}& \checkmark& \checkmark& \ding{53}& \ding{53}&Multiple\\
				\hline
				\cite{yang2021distributed} &  \checkmark& \ding{53}& \checkmark&\checkmark & \checkmark& \ding{53}&Multiple\\
				\hline
				\cite{sun2020green} &  \checkmark& \ding{53}&  \checkmark& \checkmark&\checkmark &\ding{53} &Multiple\\
				\hline
				\cite{evaggelia2008ac} & \ding{53}&  \checkmark& \checkmark& \ding{53}& \ding{53}&\checkmark &Zero\\
				\hline
				\cite{shi2016ac} & \checkmark&  \checkmark& \checkmark& \ding{53}& \ding{53}& \checkmark&Zero\\
				\hline
				\cite{zhao2015uac} & \ding{53}&  \checkmark& \checkmark& \ding{53}& \ding{53}&\checkmark &Zero\\
				\hline
				\cite{lid2021uac} & \ding{53}&  \checkmark&\checkmark&\checkmark & \ding{53}& \checkmark&Single\\
				\hline
				This work& \checkmark&  \checkmark&  \checkmark& \checkmark& \checkmark& \checkmark& Multiple\\
				\hline
		\end{tabular}}
		\label{table_MAP}
	\end{table*}

	All the above works \cite{huang2019reconfigurable,hua2021greenris,zhou2020hardware,wang2020active,yang2021distributed,sun2020green} for green cellular networks, however, do not consider the infeasibility issue when the QoS constraints cannot be satisfied by all users.
	User admission control is a powerful solution to address the infeasibility problem by selecting a subset of users for communication.
	In \cite{evaggelia2008ac}, the authors proposed two computationally efficient algorithms to  jointly optimize the multi-user transmit beamforming and the active user set.
	By employing a sparse optimization framework, the authors in  \cite{shi2016ac} considered the network power minimization and user scheduling with multicast beamforming.	
	In \cite{zhao2015uac}, the authors dealt with the admission control problem by converting the problem into its dual downlink form and then reformulating the problem as a sparsity-maximization problem.
	The aforementioned studies \cite{evaggelia2008ac,shi2016ac,zhao2015uac} only focused on the user admission control problem for conventional multi-user transmission without  RIS.
	For RIS-assisted user admission control, the author in \cite{lid2021uac} considered maximizing the number of users that can be served with certain QoS requirements by optimizing the transmit beamforming and phase shifts of RIS.
	However, the work in \cite{lid2021uac} 
	only considered one RIS, which did not enjoy the degree of freedom by controlling the on-off status of multiple RISs. 
	To the best of our knowledge, this is the first work that jointly considers RIS scheduling and user admission control to minimize the power consumption of the BS and RISs in a multi-RIS assisted multi-user MISO BC.
	All the existing works do not address the infeasibility issue when the QoS constraints of  users are not satisfied in RIS empowered wireless networks.
	\subsection{Contributions}
	In Table \ref{table1}, we summarize the existing works for RIS-assisted green cellular networks, and illustrate our major contributions. 
	In this work, motivated by the limitation of the existing works, we fill in the research gap and deal with the infeasibility issue by jointly considering the power control and user admission control problems in order to guarantee the feasibility of the QoS constraints for all users as well as the transmit power constraint. 
	A unified optimization framework is proposed to jointly optimize the beamforming design of the BS, the phase shifts of RISs, as well as the active RIS/user set. The major contributions of this work are summarized as follows:
	\begin{itemize}
		\item We consider a MISO BC with multiple RISs and multiple users.
		We first formulate the joint transmit beamforming vectors, active RIS set, and phase shifts of RISs optimization problem to minimize the power consumption of the BS and RISs subject to the QoS constraints of users and the transmit power constraint of the BS.
		As the problem would encounter the infeasibility issue when the QoS constraints cannot be guaranteed by all users for a given power constraint of the BS, we then propose to solve a user admission control problem to address such an issue. Specifically, we formulate a joint transmit beamforming, RIS phase shifts, and admitted user set optimization problem to maximize the number of users that can be served under the QoS constraint of each admitted user and the transmit power constraint. 
		\item In order to minimize the network power consumption, we design an alternating optimization (AO) framework to iteratively optimize the transmit beamforming vectors, active RIS set, and phase shifts of RISs.
		In particular, we first decompose the original bi-quadratically constrained quadratic power minimization problem into two subproblems with non-convex quadratic constraints, i.e., a convex subproblem for transmit beamforming design, and a mixed combinatorial optimization subproblem for the joint optimization of the active RIS set and phase shifts of RISs. 
		The former is solved optimally while the latter is solved by the proposed low-complexity difference-of-convex (DC) programming framework. 
		\item Besides, we propose a low-complexity zero-forcing (ZF)-based algorithm to design the transmit beamforming vectors and the phase shifts of RISs. With a much lower computational complexity than the DC-based algorithm, the proposed ZF-based algorithm acts as an alternative approach in some practical delay-sensitive scenarios.
	\end{itemize}
	
	We conduct extensive simulations to evaluate the effectiveness of the proposed algorithms. Simulation results show that the proposed approaches successfully address the aforementioned infeasibility issue, and outperform the existing convex approximation approaches in terms of network power efficiency and user transmission rate in various simulation settings.
	\subsection{Organization and Notations}
	\emph{Organization: }The remainder of the paper is organized as follows.
	Section \ref{Sec_SysMod} presents the system model and problem formulation.
	Section \ref{Sec_NPM} and \ref{Sec_UAC} respectively specify the proposed optimization framework for the network power minimization and user admission control.
	Simulation results are shown in Section \ref{Sec_NumRes}. 
	Finally, Section \ref{Sec_Con} concludes the paper.
	
	\emph{Notations: }
	The symbol $\mathbb{E}(\cdot)$ denotes the statistical expectation and $\mathbb{R}^+$ denotes the non-negative real domain.
	The complex normal distribution is denoted as $\mathcal{C}\mathcal{N}(\mu,\sigma^2)$.
	The symbol $|\cdot|$, $||\cdot||$ and $||\cdot||_0$ denote the modulus, spectral norm and $l_0$ norm of a vector. 
	The transpose, conjugate transpose, trace operator, rank operator, Moore-Penrose pseudoinverse and diagonal matrix are denoted as $(\cdot)^{\sf T}$, $(\cdot)^{\sf H}$, $\mathrm{Tr}(\cdot)$, $\mathrm{Rank}(\cdot)$, $|\cdot|^{\dagger}$ and $\mathrm{diag}(\cdot)$ respectively.
	Notation $[\bm{x}]_{(1:N)}$ denotes the first $N$ elements of vector $\bm{x}$. 
	The identity matrix is denoted by $\bm{I}$.
	The symbol $\text{card}(\cdot) $ represents the cardinality of a given set. 
	The symbol $\partial_{\{\cdot\}}$ denotes the subgradient operation, and $\langle\bm{X}, \bm{Y}\rangle$ is the inner product of matrices $\bm{X}$ and $\bm{Y}$.
	Finally, $\Re(\cdot)$ denotes the real part of a complex number.

	\section{System Model and Problem Formulation} \label{Sec_SysMod}
	In this section, we first introduce the system model of a multi-RIS-assisted multi-user MISO downlink cellular network. 
	Then, we present the problem formulation of the network power minimization and user admission control.
	\subsection{System Model}
	We consider a downlink transmission model consisting of one $M$-antenna BS, $K$ single-antenna users, and $L$ distributed RISs. { In particular, the RISs are deployed to enhance the communication between the BS and users.}
	The users and RISs are respectively indexed by $\mathcal{K} = \{ 1, 2, \ldots, K\}$ and $\mathcal{L} = \{ 1, 2, \ldots, L \}$.
	
	Let $s_k \in \mathbb{C}$ and $\bm{w}_k\in \mathbb{C}^{M \times 1}$  denote the corresponding transmit symbol and beamforming vector at the BS for users $k$, respectively.
	The signal transmitted by the BS is expressed as
	\begin{eqnarray}\label{transmitted_signal}
		\bm x  = \sum_{k\in\mathcal{K}} \bm{w}_k s_k,
	\end{eqnarray}
	
	{Assume that user symbols $\{s_k|k \in \mathcal{K}\}$ are user-independent 
	and the transmit power of the BS is constrained by }
	\begin{eqnarray}\label{maximum transmit power constraint}
		\sum_{k\in \mathcal{K}} \| \bm{w}_k \|^2 \le P_{\mathrm{max}}.
	\end{eqnarray}

	%
	On-off RIS control is considered in this work. In particular, the power consumed by the active RISs is considered as part of the network power consumption, while an inactive RIS is switched off without consuming any power. A detailed power consumption model is specified in Subsection \ref{sec:power_model}.
	To minimize the network power consumption, we propose to dynamically switch off some RISs while satisfying the QoS requirements of all users.
	Let $\mathcal{A}\subseteq \mathcal{L}$ and $\mathcal{Z}\subseteq \mathcal{L}$ be the index set of active and inactive RISs, respectively (note that $\mathcal{A} \cup \mathcal{Z} = \mathcal{L}$).
	For an active RIS $l\in \mathcal{A}$, the phase-shift matrix can be optimized through a diagonal matrix $\bm \Theta_l = \mathrm{diag}(\theta_{l,1}, \theta_{l,2}, \ldots, \theta_{l,N_l})\in \mathbb{C}^{N_l \times N_l}$, where $\theta_{l,n}\in \mathbb{C}$ denotes the reflection coefficient of passive element $n$ equipped on RIS $l$ with $n\in \mathcal{N}_l=\{1,2,\cdots,N_l\}$.
	
	Let $\bm h_{d,k} \in \mathbb{C}^{M \times 1}$, $\bm T_l \in \mathbb{C}^{N_l \times M}$, $\bm h_{l,k} \in \mathbb{C}^{N_l \times 1},\  \forall \, l \in \mathcal{L}, \forall \, k \in \mathcal{K}$ denote the channel responses from the BS to user $k$, from the BS to RIS $l$, and from RIS $l$ to user $k$, respectively. 
	Although it is generally difficult to obtain perfect CSI, various efficient channel estimation methods proposed for RIS-assisted wireless networks can be adopted to provide accurate CSI of the
	channels, e.g., the alternating least squares method \cite{wei2021tc} and the compressive-sensing
	based method \cite{liu2020jsac}. 
	Meanwhile, imperfect CSI in RIS-assisted networks has been widely studied, and the influence of CSI imperfection has been widely known from current literature \cite{zhou2020robust,lu2020wcl,han2019tvt}. Such observation also holds for the system model considered in this work. 
	With the assistance of active RISs in set $\mathcal{A}$, the received signal $y_k\in \mathbb{C}$ at user $k$ is given as
	\begin{eqnarray}\label{received_signal}
		y_k = \left( \sum_{l \in \mathcal{A}} \bm h_{l,k}^{\sf{H}} \bm \Theta_l \bm T_l +\bm h_{d,k}^{\sf{H}}  \right) \sum_{i\in \mathcal{K}} \bm{w}_i s_i + z_k, \forall \, k\in \mathcal{K}, 
	\end{eqnarray}
	where $z_k \sim \mathcal{CN}(0, \sigma_k^2)$ is the additive white Gaussian noise (AWGN) at user $k$. 
	Due to the severe path loss, the power of the signals that are reflected by the RIS with two or more times is assumed to be negligible \cite{wu2018intelligent}.
	
	We assume user-independent detection.
	Based on (\ref{transmitted_signal}) and (\ref{received_signal}), the achievable signal-to-interference-plus-noise ratio (SINR) $\gamma_k$ for user $k$ is computed as
	\begin{align} \label{SINR_k}
		\gamma_{k} \left(\mathcal{A},\bm{w}_k,\{ \bm \Theta_l \}\right) =
		\frac{\left|\left( \sum_{l \in \mathcal{A}} \bm h_{l,k}^{\sf{H}} \bm \Theta_l \bm T_l + \bm h_{d,k}^{\sf{H}} \right) \bm{w}_k \right|^2}{\sum_{j \neq k} \left|\left( \sum_{l \in \mathcal{A}} \bm h_{l,k}^{\sf{H}} \bm \Theta_l \bm T_l + \bm h_{d,k}^{\sf{H}} \right)\bm{w}_j \right|^2 \!\!\!+ \! \sigma_k^2}. 
	\end{align}
	Each user $k$ in $\mathcal{K}$ has the following QoS constraint
	\begin{eqnarray}\label{QoS_con}
		\gamma_k\left(\mathcal{A},\bm{w}_k,\{ \bm \Theta_l \}\right) \geq \gamma_k^{\textrm{th}}, \forall \, k \in \mathcal{K},
	\end{eqnarray}
	where $\gamma_k^{\textrm{th}}$ denotes the SINR threshold of $s_k$. 

	\subsection{Power Consumption Model}\label{sec:power_model}
	The power consumption introduced by active RISs cannot be simply ignored since they are usually densely deployed and each has a hardware static power consumption. 
	Therefore, to design a power-efficient transmission network, we consider the following network power consumption model which takes the power consumption of both the BS and RISs into account.
	Specifically, the network power consumption of the considered multi-RIS-assisted system includes the transmit power of the BS, the total hardware static power consumption of the BS, and the total hardware static power consumption of all active RISs.
	%
	%
	We consider an empirical linear power consumption of the total cellular network, which is given by
	\begin{eqnarray} \label{Eq_NPC}
		{P}_{\text{total}}(\mathcal{A}, \{ \bm{w}_k \}) = \sum_{k\in \mathcal{K}} \frac{1}{\eta} \| \bm{w}_k \|^2 + P_{\mathrm{static}}+ \sum_{l \in \mathcal{A}}P_{\mathrm{RIS}}(N_l),
	\end{eqnarray}
	where $\eta$ is the drain efficiency of the radio frequency power amplifier at the BS with a typical value as 60$\%$\cite{huang2019reconfigurable}, 
	$P_{\mathrm{static}}$ is the total static power consumption of the BS and RISs (e.g., power consumption of
	the master control board and drive circuits of RIS \cite{wang2023arxiv}), and $P_{\mathrm{RIS}}(N_l)$ is the total power consumption of the elements in RIS $l$, where $N_l$ represents the number of elements in RIS $l$. 
	Note that the power consumption of RISs is mainly due to the tuning of the RIS reflecting elements. 
	The power consumption of RIS $l$ can be expressed as $P_{\mathrm{RIS}}(N_l) = N_l P_{\mathrm{RE}},\forall \, l \in \mathcal{A}$, where $P_{\mathrm{RE}}$ denotes the power consumption of each phase shift.
	
	Since the static power consumption of the BS and RISs $P_{\mathrm{static}}$ is constant and does not vary with the optimization variables, we can convert the network power consumption ${P}_{\text{total}}(\mathcal{A}, \{ \bm{w}_k \})$ in (\ref{Eq_NPC}) to the following simplified model
	\begin{eqnarray} \label{Eq_NPC1}
		P(\mathcal{A}, \{ \bm{w}_k \}) = \sum_{k \in \mathcal{K}} \frac{1}{\eta} \| \bm{w}_k \|^2 + \sum_{l \in \mathcal{A}}P_{\mathrm{RIS}}(N_l).
	\end{eqnarray}
	
	Note that some limitations of power amplifier (PA) in RIS-assisted wireless networks need to be addressed.
	Due to the nonlinear property of PA, there exists a nonlinear operation region where the amplification gain is nonlinear but the phase is almost perfectly preserved. 
	Therefore, the amplifier should be properly deployed with the power of the attenuated incident signal in its linear operation region \cite{long2021twc}. 
	Besides, some distortion reduction techniques (e.g., digital pre-distortion \cite{john2017tmtt}) can be employed for practical implementations.

	\subsection{Network Power Minimization Problem}
	Given the considered system model in (\ref{maximum transmit power constraint}), (\ref{QoS_con}) and (\ref{Eq_NPC1}), we aim at jointly optimizing the set of active RISs $\mathcal{A}$, the transmit beamformers $\{\bm{w}_k\}$, and the reflection coefficient matrix of the active RISs $\{\mathbf{\Theta}_l\}$ to minimize the total network power consumption $P(\mathcal{A}, \{ \bm{w}_k \})$, under the minimum QoS constraints of all users
	and the transmit power constraint of the BS.
	The formulated power minimization problem is given as
	\begin{subequations} \label{OP_Original}
		\begin{align}
			\hspace{-3mm} \mathscr{P}: \mathop{\text{minimize}}_{ \mathcal{A}, \{ \bm{w}_k \}, \{ \bm \Theta_l \}} \label{P_Obj}\hspace{2mm}  & P(\mathcal{A}, \{\bm{w}_k \})  \\
			\text{subject to} 
			\label{P_SINR} \hspace{2mm}  & \gamma_k(\mathcal{A},\bm{w}_k,\{ \bm \Theta_l \}) \ge \gamma_k^{\textrm{th}}, \forall \, k \in \mathcal{K},  \\
			\label{P_power} \hspace{2mm}  & \sum_{k\in \mathcal{K}} \| \bm{w}_k \|^2 \le P_{\mathrm{max}},  \\
			\label{P_unit} \hspace{2mm}  & | \theta_{l,n}| = 1,\ \forall l \in \mathcal{A},\ \forall n \in \mathcal{N}_l,
		\end{align}
	\end{subequations}
	where (\ref{P_unit}) is the unit-modulus constraint of each phase shift at the active RISs.
	
	Solving problem $\mathscr{P}$ is highly challenging due to the following three aspects: 1) the discrete RIS set variable $\mathcal{A}$;
	2) the non-convex SINR constraints (\ref{P_SINR});
	3) the non-convex modulus constraints (\ref{P_unit}). 
	In Section III, the proposed optimization framework to solve the problem is delineated. 
	
	{Besides the aforementioned optimization issues, problem $\mathscr{P}$  also encounters the infeasibility issue when the transmit power constraint of the BS is not large enough to support the QoS constraints $\gamma_k^{th}$ for $k\in\mathcal{K}$.} In the next subsection, we discuss the approach we proposed to deal with the problem when infeasibility arises.    
	
	\subsection{User Admission Control Problem}
	When problem $\mathscr{P}$ in (\ref{OP_Original}) becomes infeasible, either constraint (\ref{OP_Original}b) or (\ref{OP_Original}c) cannot be satisfied. An alternative approach is to select a subset of active users which satisfy both constraints. 
	Such an approach is also known as user admission control.
	{Specifically, we propose to maximize the number of active users that can be simultaneously served subject to the QoS constraint of each admitted user and the transmit power constraint of the BS.}
	In such setting,  all RISs are turned on to support the signal transmission of more active users.
	Denote the set of admitted users as $\mathcal{S}$ with $\mathcal{S}\subseteq \mathcal{K}$, the formulated admission control problem is given as
	\begin{subequations} \label{user_admission_control}
		\begin{align}
			\hspace{-3mm} \mathscr{U}: \mathop{\text{maximize}}_{ \{ \bm{w}_k \}, \{ \bm \Theta_l \}, \mathcal{S}} \label{P_Obj_s}\hspace{3mm}  & \text{card}(\mathcal{S})  \\
			\text{subject to}\ \ 
			\label{P_SINR_s} \hspace{3mm}  & \gamma_k(\mathcal{L},\bm{w}_k,\{ \bm \Theta_l \}) \ge \gamma_k^{\textrm{th}}, \forall \, k \in \mathcal{S},  \\
			\label{P_power_s} \hspace{3mm}  & \sum_{k\in \mathcal{S}} \| \bm{w}_k \|^2 \le P_{\mathrm{max}},  \\
			\label{P_unit_s} \hspace{3mm}  & | \theta_{l,n}| = 1,\ \forall \, l \in \mathcal{L},\ \forall \, n \in \mathcal{N}_l.
		\end{align}
	\end{subequations}
	Once the admitted user set $\mathcal{S}$ is determined by (\ref{user_admission_control}), the network power minimization problem could be further considered to enhance the network power control.
	Again, problem (\ref{user_admission_control}) is challenging to solve due to the combinatorial nature of the objective function (\ref{user_admission_control}a) as well as the non-convex constraints (\ref{user_admission_control}b) and (\ref{user_admission_control}d).
	
	To tackle the challenges in (\ref{OP_Original}) and (\ref{user_admission_control}), we propose a unified optimization framework to handle the two problems based on an AO approach, as delineated in Section III and Section IV.
	For the network power minimization problem, we propose to optimize $\mathcal{A}$, $\{ \bm{w}_k \}$, and $\{ \bm \Theta_l \}$ by alternatively solving a second-order cone programming (SOCP) problem and a mixed combinatorial optimization problem, which is presented in Section III.
	For the user admission control problem, we apply the same optimization flow to jointly optimize $\mathcal{S}$, $\{ \bm{w}_k \}$, and $\{ \bm \Theta_l \}$ in Section IV.
	
	\section{Alternating Framework for Network Power Minimization}\label{Sec_NPM}
	It is challenging to solve the network minimization problem $\mathscr{P}$ due to the combinatorial RIS set selection, the SINR constraint (\ref{OP_Original}b), and the non-convex constant modulus constraint (\ref{OP_Original}d).
	Recently, lots of existing works, e.g., \cite{huang2019reconfigurable,guo2020weighted} applied a powerful alternating optimization to design the active beamforming at the BS and the phase-shift matrix at the RIS,
	which always yields convex constraints given the phase-shift
	matrix (i.e., affine constraints \cite{huang2019reconfigurable} and quadratic constraints \cite{guo2020weighted}). 
    In this subsection, to decouple the optimization variables, we also employ an AO-based approach to iteratively optimize $\{\bm{w}_k\}$, and jointly optimize $\mathcal{A}$ and $\{\bm{\Theta}_l\}$. 
	Specifically, we first optimize $\{\bm{w}_k\}$ with given $\mathcal{A}$ and the corresponding RIS phase-shift matrices $\{\bm{\Theta}_l\}$ (in Section \ref{sec:3_a}).
	Then we jointly optimize $\mathcal{A}$ and $\{\bm{\Theta}_l\}$ by solving a mixed combinatorial optimization problem with given $\{\bm{w}_k\}$ (in Section \ref{sec:3_b}).
	To further improve the performance, we transform the mixed combinatorial optimization problem into a DC programming by using the binary relaxation and DC method (in Section \ref{sec:3_c}). 

	\subsection{Beamforming Optimization}
	\label{sec:3_a}
	Given the active RIS set $\mathcal{A}$ and the corresponding RIS phase-shift matrices $\{ \bm \Theta_l \}$, the transmit beamformers $\{\bm{w}_k\}$ are optimized by solving the network power minimization problem $\mathscr{P}(\{\bm{w}_k\})$ as detailed in this subsection.
	
	The concatenated channel response $\bm{{h}}_{k}^{\sf{H}} (\mathcal{A},\{ \bm \Theta_l \} )\in \mathbb{C}^{1\times M}$ of user $k$ is defined as
	\begin{eqnarray} \label{Eq_CompH}		
		\bm{{h}}_{k}^{\sf{H}} \big(\mathcal{A},\{ \bm \Theta_l \}\big) = \sum_{l \in \mathcal{A}} \bm{h}_{l,k}^{\sf{H}} \bm{\Theta}_{l} \bm{T}_{l} + \bm h_{d,k}^{\sf{H}},\ \forall k\in \mathcal{K},
	\end{eqnarray}
	which is fixed for a given $\mathcal{A}$ and $\{\bm{\Theta}_l\}$.
	Due to the fact that an arbitrary phase rotation of vector $\bm{w}$ does not affect the SINR constraint (\ref{P_SINR}) \cite{ami2006tsp}, we reformulate constraint (\ref{P_SINR}) as a second-order cone (SOC) constraint, which is given by 
	\begin{eqnarray}\label{Eq_SINRSOC}
		\mathcal{C}_1(\{ \bm{w}_k  \}):   \sqrt{\sum_{j \neq k} \frac{1}{\sigma_k^{2}}\left|\bm{{h}}_{k}^{\sf{H}}\big(\mathcal{A},\{ \bm \Theta_l \}\big) \bm{w}_j\right|^{2} \!+ \! 1} \leq \nonumber\\
		\frac{1}{\sqrt{\gamma_k^{\textrm{th}} \sigma_k^{2}}} \Re\left(\bm{{h}}_{k}^{\sf{H}}\big(\mathcal{A},\{ \bm \Theta_l \}\big) \bm{w}_k\right),\forall \, k \in \mathcal{K}.
	\end{eqnarray}
	
	Given the active RIS set $\mathcal{A}$ and the corresponding phase-shift matrices $\{ \bm \Theta_l \}$, the network power minimization problem in (\ref{OP_Original}) is equivalently transformed into the following problem $\mathscr{P}(\{\bm{w}_k\})$
	\begin{subequations} \label{OP_Omega}
		\begin{align}
			\mathscr{P}\big(\{ \bm{w}_k\} \big):\  \mathop{\text{minimize}}_{ \{ \bm{w}_k \}} 
			\hspace{2mm}& \sum_{k \in \mathcal{K}} \frac{1}{\eta} \| \bm{w}_k \|^2  \\
			\text{subject to} 
			\hspace{2mm}&  \mathcal{C}_1(\{ \bm{w}_k\} ), (\ref{P_power}).
		\end{align}
	\end{subequations}
	$\mathscr{P}(\{\bm{w}_k\})$ is a convex SOCP problem, which can be efficiently solved by using CVX \cite{cvx}. 
	
	\subsection{Joint Optimization for Active RIS Set and Phase Shifts}
	\label{sec:3_b}
	For given beamforming vectors, we then
	optimize the active RIS set $\mathcal{A}$ and the corresponding RIS phase-shift matrices $\{ \bm \Theta_l \}$.
	A naive method is to search over all possible active RIS sets.
	However, this method suffers from the practicability and scalability issues due to the exponential running time, e.g., it needs to check $2^L$ active RIS sets in the worst case.
	To reduce the computational complexity, we first leverage the semidefinite relaxation (SDR) technique, then we further improve the SDR in terms of its performance by proposing the DC algorithm in the next subsection.
	
	For given beamforming vectors $\{ \bm{w}_k \}$, problem $\mathscr{P}$ becomes
	\begin{subequations}  \label{OP_Theta}
		\begin{align}
			\mathscr{P}(\mathcal{A},\{ \bm \Theta_l \})\!\!:\mathop{\text{minimize}}_{ \mathcal{A},\{ \bm \Theta_l \}} \hspace{2mm}  & \sum_{l \in \mathcal{A}} P_{\mathrm{RIS}}(N_l)  \hfill\\
			\text{subject to} 
			\hspace{2mm} & \gamma_k (\mathcal{A},\{ \bm \Theta_l \})  \ge \gamma_k^{\textrm{th}}, \forall \, k \in \mathcal{K}.  
		\end{align}
	\end{subequations}

	To solve the combinatorial and non-convex problem (\ref{OP_Theta}), we first introduce an auxiliary variable vector $\bm{\beta}=\{\beta_1,\beta_2,\ldots,\beta_L\}$ to represent the indicators of active RISs, and each element $\beta_l \in \{0,1\}$ is a binary variable defined as
	\begin{equation}
		\beta_l=\left\{
		\begin{aligned}
			0 & , & \text{RIS}\  l\  \text{is inactive},  \\
			1 & , & \text{RIS}\  l\  \text{is active}.\ \
		\end{aligned}
		\right.
	\end{equation}
	%
	Let $\tilde{\theta}_{l,n} = \beta_l \theta_{l,n}$ and $\bm{\tilde{\theta}}_l = [\tilde{\theta}_{l,1},\tilde{\theta}_{l,2},\ldots,\tilde{\theta}_{l,N_l}]^{\sf{H}} \in \mathbb{C}^{N_l \times 1},\ \forall \, l \in \mathcal{L}$.
	We denote $\bm a_{k,j}(l)=\mathrm {diag}(\bm h_{l,k}^{\sf{H}})\bm T_l \bm{w}_j$, $\forall \, k, j \in \mathcal{K},\ \forall \, l \in \mathcal{L}$ and $ b_{k,j}=\bm h_{d,k}^{\sf{H}}\bm{w}_j $.
	Then, we have
	\begin{align}\label{equ:a}
		(\beta_l \bm h_{l,k}^{\sf{H}} \bm \Theta_l \bm T_l +\bm h_{d,k}^{\sf{H}})\bm{w}_j = \bm \tilde{\bm{\theta}}_l^{\sf{H}}\bm a_{k,j}(l)+b_{k,j} , \forall \, k, j \in \mathcal{K}, \forall \, l \in \mathcal{L}.
	\end{align}
	Based on the above auxiliary variables, the SINR constraints (\ref{P_SINR}) can be rewritten as
	\begin{eqnarray} \label{Eq_SINRk}
		\hspace{-16mm} & \gamma_k^{\textrm{th}} \left (\sum_{j \neq k} \Big| \sum_{l \in \mathcal{L}} \bm \tilde{\bm{\theta}}_l^{\sf{H}}\bm a_{k,j}(l) + b_{k,j} \Big|^2 + \sigma_k^2 \right) 
		\leq \nonumber\\
		&\left| \sum_{l \in \mathcal{L}} \bm \tilde{\bm{\theta}}_l^{\sf{H}}\bm a_{k,k}(l) + b_{k,k} \right|^2, \forall \, k\in \mathcal{K}.
	\end{eqnarray}
	The SINR constraint (\ref{Eq_SINRk}) is still non-convex.
	We combine $\bm{\tilde{a}}_{k,j}=[\bm a_{k,j}(1)^{\sf{H}},\bm a_{k,j}(2)^{\sf{H}},\dots,\bm a_{k,j}(L)^{\sf{H}}]^{\sf{H}}$ $\in \mathbb{C}^{{{N}_{\text{total}}} \times 1},\forall k,j \in \mathcal{K}$ and $\bm{\tilde{\theta}}=[\tilde{\bm{\theta}}_1^{\sf{H}},\tilde{\bm{\theta}}_2^{\sf{H}},\dots,\tilde{\bm{\theta}}_L^{\sf{H}}]^{\sf{H}} \in \mathbb{C}^{{{N}_{\text{total}}}\times 1}$, where ${{N}_{\text{total}}=\sum_{l \in \mathcal{L}}N_l}$ denotes the total number of passive reflecting elements at all RISs. 
	Consequently, we have $\sum_{l \in \mathcal{L}}\bm{\tilde{\theta}}_l^{\sf{H}}\bm a_{k,j}(l)=\bm{\tilde{\theta}}^{\sf{H}}{\bm{\tilde{a}}}_{k,j},\ \forall \, k, j \in \mathcal{K}$. 
	By further extending the cardinality, we can define the following variables
	\begin{eqnarray}
		\boldsymbol{R}_{k,j}=
		\left[
		\begin{array}{ll}{{\bm {\tilde{a}}}_{k,j} \bm {\tilde{a}}^{\sf{H}}_{k,j}} & {{\bm {\tilde{a}}_{k,j}} b_{k,j}^{\sf{H}}} \\ 
			{b_{k,j}{\bm {\tilde{a}}}_{k,j}^{\sf{H}}}
			& {0}\end{array}\right] \in \mathbb{C}^{({{N}_{\text{total}}}+1) \times ({{N}_{\text{total}}}+1)},
	\end{eqnarray}
	\begin{eqnarray}
		\bm{\hat{\theta}}=\left[\begin{array}{l}\bm{\tilde{\theta}} \\ {1}
		\end{array}
		\right]\in \mathbb{C}^{({{N}_{\text{total}}}+1) \times 1}.
	\end{eqnarray}
	Thereby, the non-convex constraint (\ref{Eq_SINRk}) is equivalently transformed into
	\begin{eqnarray}
		\mathcal{C}_2(\mathcal{L},\bm{\hat{\theta}}): \hspace{-6mm}  \gamma_k^{\textrm{th}} \left(\sum_{j \neq k} \left( \bm{\hat{\theta}}^{\sf{H}}\bm R_{k,j} \bm{\hat{\theta}} + |b_{k,j}|^2\right)  + \sigma_k^2 \right)  \leq \\ \bm{\hat{\theta}}^{\sf{H}}\bm R_{k,k} \bm{\hat{\theta}} + |b_{k,k}|^2, \forall \, k \in \mathcal{K}.
	\end{eqnarray}
	
	Based on the above introduced auxiliary variables, problem $\mathscr{P}(\mathcal{A},\{ \bm \Theta_l \})$ can be simplified to the following mixed-integer and quadratic non-convex optimization problem
	\begin{subequations} \label{OP_act}
		\begin{align}
			\mathop{\text{minimize}}_{ \bm \beta,  \bm{\hat{\theta}} }  \hspace{3mm} &\sum_{l \in \mathcal{L}} \beta_l P_{\mathrm{RIS}}(N_l)  \\
			\text{subject to} \hspace{6mm}
			\label{OP_act01}  \hspace{-3mm} & \mathcal{C}_2(\mathcal{L},\bm{\hat{\theta}}),  \\
			\label{OP_act_q1}\hspace{-3mm}  & |\bm{\hat{\theta}}_i| = \beta_l, \forall \, l \in \mathcal{L}, \forall i \in \mathcal{I}_l,\\ &|\bm{\hat{\theta}}_{N_{\text{total}}+1}| = 1, \\
			\label{OP_act_a} \hspace{-3mm} & \beta_l \in\{ 0,1\},\forall \, l \in \mathcal{L}, 
		\end{align}
	\end{subequations}
	where $\mathcal{I}_l = \{ \sum_{j=1}^{l-1}N_j + 1, \dots, \sum_{j=1}^{l}N_j \}$ denotes the passive reflecting elements set of RIS $l,\ \forall \, l \in \mathcal{L}$.
	If we obtain a feasible solution, denoted as $\bm{\hat{\theta}}^*$, by solving problem (\ref{OP_act}), then a feasible solution for the phase-shift vector can immediately be obtained by setting $\bm{\tilde{\theta}}^*\!=\![\bm{\hat{\theta}}^*\!/\!\bm{\hat{\theta}}^*_{{{N}_{\text{total}}}+1}]_{(1:{{N}_{\text{total}}})}$.
	To deal with the non-convex binary constraints (\ref{OP_act_a}), we relax the constraints (\ref{OP_act_a}) into convex unit interval constraints $\{0\leq \beta_l \leq 1,\ \forall \, l  \in \mathcal{L}\}$, which can then be solved by the proposed DC algorithm \cite{lee2011mixed}.
	Thus, problem (\ref{OP_act}) is relaxed to the following homogeneous non-convex quadratic constrained quadratic programs (QCQP) problem \cite{So2007On}
	\begin{subequations} \label{OP_NC}
		\begin{align}
			\mathop{\text{minimize}}_{ \bm \beta,  \bm{\hat{\theta}}}  \hspace{3mm} &\sum_{l \in 	\mathcal{L}} \beta_l P_{\mathrm{RIS}}(N_l)   \\
			\text{subject to} \hspace{6mm}
			\hspace{-3mm} &(\ref{OP_act}b), (\ref{OP_act}c), (\ref{OP_act}d),\\ 
			\hspace{-3mm}  & 0\leq \beta_l \leq 1, \forall \, l  \in \mathcal{L}. 
		\end{align}
	\end{subequations}
	
	\textbf{Discussion: }In this part, we provide some discussions on the impact of the solution to the above problems.
	For problem (\ref{OP_NC}), the successive convex approximation (SCA) method can be employed by relaxing the original problem as a SOCP problem.
	However, the SCA method may result in a suboptimal solution with performance degradation. 
	To further improve the performance, the SDR technique \cite{luo2010sdr} can be exploited by formulating the optimization problem as
	a semidefinite programming (SDP) form by lifting vector $\bm{\hat{\theta}}$ into a positive semidefinite (PSD) matrix $\bm{\hat{\Theta}}= {\bm{\hat{\theta}}} {\bm{\hat{\theta}}}^{\sf{H}} \in \mathbb{C}^{({{N}_{\text{total}}}+1) \times ({{N}_{\text{total}}}+1)}$  with $\bm{\hat{\Theta}} \succeq 0$ and ${\rm {rank}}(\bm{\hat{\Theta}})=1$.
	Let $\mathrm{Tr}(\bm R_{k,j} \bm{\hat{\Theta}}) = {\bm{\hat{\theta}}}^{\sf{H}}\bm R_{k,j} {\bm{\hat{\theta}}}$, the QCQP problem (\ref{OP_NC}) can be reformulated into a rank-one constrained  optimization problem, which is given by
	\begin{subequations} \label{OP_The_rank_one}
		\begin{align}
			\mathop{\text{minimize}}_{{\bm \beta}, \bm{\hat{\Theta}}} 
			\hspace{3mm} & \sum_{l  \in \mathcal{L}} {\beta}_l P_{\mathrm{RIS}}(N_l)  \\
			\text{subject to} \hspace{6mm}
			\hspace{-3mm}&\mathcal{C}_3(\mathcal{L},\hat{\bm{\Theta}}),\label{equ:C3}\\
			\hspace{-3mm}	& \bm{\hat{\Theta}}(i,i) = {\beta}_l^2,\\
			& \bm{\hat{\Theta}}({{N}_{\text{total}}}+1,{{N}_{\text{total}}}+1) = 1, \bm{\hat{\Theta}} \succeq 0,\\
			\hspace{-3mm} 	&\text{rank}(\bm{\hat{\Theta}})=1,\label{equ:rank1}\\
			\hspace{-3mm} 	&0\leq {\beta}_l \leq 1, \forall \, l \in \mathcal{L},
		\end{align}
	\end{subequations}
	where the constraint (\ref{equ:C3}) is defined as
	\begin{align}
		\mathcal{C}_3(\mathcal{L},\hat{\bm{\Theta}}):~\gamma_k^{\textrm{th}}\left(  \sum_{j \neq k}{\rm Tr}(\bm R_{k,j}\bm{\hat{\Theta}}) +  \sum_{j \neq k}|b_{k,j}|^2+ \sigma_k^2 \right)\leq \nonumber\\
		 {\rm Tr}(\bm R_{k,k}\bm{\hat{\Theta}}) + |b_{k,k}|^2, \forall \, k\in \mathcal{K}.
	\end{align}
	
	To solve problem (\ref{OP_The_rank_one}), the SDR approach first drops the rank-one constraint (\ref{equ:rank1}) and then solves the relaxed SDP problem \cite{luo2010sdr}. 
	The solution obtained by the SDR technique can satisfy the original problem (\ref{OP_The_rank_one}) if the solution meets the rank-one constraint $\text{rank}(\bm{\hat{\Theta}})=1$.
	However, when the dimension of the optimization variable is large, the SDR technique may fail to derive the solution satisfying the rank-one constraint \cite{chen2017tc}.
	Such an issue may yield performance degradation as suboptimal beamforming vectors and the early stopping in the procedure of alternating
	optimization.
	In wireless networks, as the number of RISs and/or passive reflecting elements increases, the performance of the SDR approach deteriorates since the probability of returning rank-one solutions is low \cite{chen2017tc}.
	
	To circumvent the limitations of SDR, in the next subsection, we propose a DC programming-based approach to accurately detect the feasibility of the non-convex constraints, thus yielding a better performance for the optimization of the active RIS set and RIS phase shifts.

	\subsection{ DC Programming Method}
	\label{sec:3_c}
	In this subsection, we first present a DC representation of the problem  (\ref{OP_The_rank_one}), then propose an efficient algorithm to solve it.
	
	\subsubsection{DC Representation for Rank-One Constraint}
	Firstly, we present an exact DC representation for the rank-one
	constraint. 
	For the matrix $\bm{\hat{\Theta}} \in \mathbb{C}^{N_{\text{total}+1}\times N_{\text{total}+1}}$, the rank-one constraint
	can be rewritten as $||\sigma_1(\bm{\hat{\Theta}}), \sigma_2(\bm{\hat{\Theta}}), \sigma_{N}(\bm{\hat{\Theta}})||_0=1$, where $\sigma_i(\bm{\hat{\Theta}})$ is the $i$-th largest singular value of matrix $\bm{\hat{\Theta}}$. 
	Note that the
	rank function is a discontinuous function. To reformulate a
	continuous function, we introduce an exact DC representation
	for the rank-one constraint as
	\begin{align}\label{equ:rank1dc}
		\text{rank}(\bm{\hat{\Theta}})=1\Leftrightarrow \text{Tr}(\bm{\hat{\Theta}})-\|\bm{\hat{\Theta}}\|_2=0.
	\end{align}		
It is noteworthy that the DC representation is a continuous
function.
	
	By adding the DC representation of the rank-one constraint into the objective function of problem (\ref{OP_The_rank_one}), we obtain
	\begin{subequations} \label{DC_OP_The}
		\begin{align}
			\mathscr{P}_{\text{DC1}}:\mathop{\text{minimize}}_{ {\bm \beta}, \bm{\hat{\Theta}}} 
			\hspace{3mm} &  \sum_{l  \in \mathcal{L}} {\beta}_l P_{\mathrm{RIS}}(N_l)+\rho\big(\text{Tr}(\bm{\hat{\Theta}})-\|\bm{\hat{\Theta}}\|_2\big) \\
			\text{subject to} \hspace{6mm}
			\hspace{-3mm}&	(\ref{OP_The_rank_one}b), (\ref{OP_The_rank_one}c), (\ref{OP_The_rank_one}d), (\ref{OP_The_rank_one}f), 
		\end{align}
	\end{subequations}
	where $\rho$ denotes the penalty parameter.
	
	\subsubsection{Difference of Strongly Convex Functions Representation}
	Although the DC program (\ref{DC_OP_The}) is non-convex,
	it has a good structure that can be exploited to develop an efficient algorithm by using SCA \cite{tao1997convex}.
	To this end, we first represent the objective function as the difference
	of two strongly convex functions and reformulate problem (\ref{DC_OP_The}) as
	\begin{align}\label{DC_theta}
		\mathop{\text{minimize}}_{ {\bm \beta}, \bm{\hat{\Theta}}} 
		&\sum_{l  \in \mathcal{L}} {\beta}_l P_{\mathrm{RIS}}(N_l)+\rho\big(\text{Tr}(\bm{\hat{\Theta}})-\|\bm{\hat{\Theta}}\|_2\big)+\Delta_{\mathcal{C}}(\bm{\hat{\Theta}}),
	\end{align}
	where $\mathcal{C}$ denotes the PSD cones satisfying the constraints in problem (\ref{DC_OP_The}) and  $\Delta_{\mathcal{C}}(\bm{\hat{\Theta}})$ is the indicator function, which is defined as
	\begin{equation}
		\Delta_{\mathcal{C}}(\bm{\hat{\Theta}})=\left\{
		\begin{aligned}
			0, \ \ & \bm{\hat{\Theta}}\in \mathcal{C},  \\
			+\infty,\ \  &\text{otherwise}.\ \
		\end{aligned}
		\right.
	\end{equation}
	It is easy to verify that problem (\ref{DC_theta}) has a special structure of minimizing the difference of two convex functions, which is given by
	\begin{eqnarray}\label{DC_1}
		\mathop{\text{minimize}}_{ {\bm \beta}, \bm{\hat{\Theta}}} \ \ 
		p(\bm{\hat{\Theta}})-q(\bm{\hat{\Theta}}),
	\end{eqnarray}
	where $p(\bm{\hat{\Theta}})$ denotes $ \sum_{l  \in \mathcal{L}} {\beta}_l P_{\mathrm{RIS}}(N_l)+\rho\text{Tr}(\bm{\hat{\Theta}})+\Delta_{\mathcal{C}}(\bm{\hat{\Theta}})+\frac{\alpha}{2}||\bm{\hat{\Theta}}||_F^2$ and $q(\bm{\hat{\Theta}})=\rho\|\bm{\hat{\Theta}}\|_2+\frac{\alpha}{2}||\bm{\hat{\Theta}}||_F^2$.
	Because of the additional quadratic terms (i.e., $\frac{\alpha}{2}||\bm{\hat{\Theta}}||_F^2$), $p(\bm{\hat{\Theta}})$ and $q(\bm{\hat{\Theta}})$are both $\alpha$-strongly convex functions.
	We can write the dual problem of (\ref{DC_1}) based on Fenchel’s duality \cite{rock2015convex} as 
	\begin{eqnarray}\label{DC_1_dual}
		\mathop{\text{minimize}}_{ {\bm \beta}, \bm{Y}} \ \ 
		q^*(\bm{Y})-p^*(\bm{Y}),
	\end{eqnarray}
	where $p^*$ and $q^*$ are the conjugate functions of $p$ and $q$, and matrix $\bm{Y}$ denotes the dual variable. 
	The conjugate function of $p$ is defined as
	\begin{eqnarray}\label{DC_2}
		p^*(\bm{Y})=\sup_{\bm{Y}\in \mathbb{C}^{(N_{\text{total}}+1)\times (N_{\text{total}}+1)}}\{\langle\bm{\hat{\Theta}},\bm{Y}\rangle-p(\bm{\hat{\Theta}}):\bm{\hat{\Theta}}\in\mathcal{X}\},
	\end{eqnarray}
	where $\langle\bm{\hat{\Theta}},\bm{Y}\rangle=\Re(\text{Tr}(\bm{\hat{\Theta}}^{\sf H}\bm{Y}))$ is computed according to Wirtinger calculus \cite{rock2015convex} in the complex domain and $\mathcal{X}$ denotes the feasible solution region of $\bm{\hat{\Theta}}$.
	
	Note that the primal and dual problems are still non-convex.
	By employing successive convex approximation \cite{tao1997convex}, we can iteratively update primal and dual variables as
	\begin{align}\label{t-iteration}
		&\bm{Y}^t= \mathop{\textrm{arg\,inf}}_{\bm{Y}}\  q^*(\bm{Y})\!-\!\big[p^*(\bm{Y}^{t-1})\!+\!\langle \bm{Y}\!\!-\!\bm{Y}^{t-1}, \bm{\hat{\Theta}}^t\rangle\big], \\
		&\bm{\hat{\Theta}}^{t+1} =\mathop{\textrm{arg\,inf}}_{ \bm{\hat{\Theta}}}\  p(\bm{\hat{\Theta}})-\big[q(\bm{\hat{\Theta}}^{t}) + \langle \bm{\hat{\Theta}}-\bm{\hat{\Theta}}^{t}, \bm{Y}^t\rangle\big].
	\end{align}
	According to Fenchel biconjugation theorem \cite{rock2015convex}, equation (\ref{t-iteration}) can be rewritten as $\bm{Y}^t  \in \partial_{\bm{\hat{\Theta}}^t}q$.
	Hence, at the $t$-th iteration, problem (\ref{DC_OP_The}) is linearized into the convex optimization problem as
	\begin{subequations} \label{t-DC_OP_The}
		\begin{align}
			\mathop{\text{minimize}}_{ {\bm \beta}, \bm{\hat{\Theta}}} 
			\hspace{3mm}&  p(\bm{\hat{\Theta}})-\langle \bm{{\hat{\Theta}}},\partial_{\bm{{\hat{\Theta}}}^{t-1}}q(\bm{\hat{\Theta}})\rangle  \\
			\text{subject to} \hspace{6mm}
				\hspace{-3mm}&(\ref{OP_The_rank_one}b), (\ref{OP_The_rank_one}c), (\ref{OP_The_rank_one}d), (\ref{OP_The_rank_one}f), 
		\end{align}
	\end{subequations}
	where $\bm{\hat{\Theta}}^{t-1}$ denotes the solution to problem (\ref{t-DC_OP_The}) at the $(t-1)$-th iteration and $\langle\bm{X}, \bm{Y}\rangle$ is computed by Wirtinger's calculus.
	Problem (\ref{t-DC_OP_The}) is convex and can be efficiently solved by CVX \cite{cvx}.
	The receiver beamforming vector $\bm{\hat{\theta}}^*$ can be obtained through Cholesky decomposition $\bm{\hat{\Theta}}^*=\bm{\hat{\theta}}^*(\bm{\hat{\theta}}^*)^H$\cite{higham1990analysis}.
	Note that the initial solution of the DC program may not be feasible due to the rank-one constraint (23).
	The initial feasible solution of the DC program could be derived by the SDR algorithm, which recovers the rank-one solution by the Gaussian randomization technique \cite{luo2020sdr}.
	Moreover, the solution of $\bm{\beta}$ in (32) may not be binary.
	To recover the binary $\bm{\beta}$, we first sort the relaxed ${\bm \beta}$, then iteratively add the RIS to the active RIS set according to the descent order of $\bm{\beta}$.

\begin{algorithm}[t]
	\caption{Proposed Algorithm for Network Power Minimization}\label{Alg}
	\KwIn{$\{\bm \Theta_l^0 \}$ and $\epsilon$.}
	\For{$t=1,2,\dots$}
	{   
		Given $\{\bm \Theta_l^{t-1} \}$, solve 	problem (\ref{OP_Omega}) to obtain $\{\bm{w}_k ^t\}$. \\
		Given $\{\bm{w}_k^t \}$, solve problem (\ref{DC_OP_The}) to obtain ${\bm \beta}^t$ and $\bm{\hat{\Theta}}^t$.\\
		\For{$i=1,2,...$}{
			Compute a subgradient of $\partial_{\hat{\bm{\Theta}}^{i-1}}\|\hat{\bm{\Theta}}\|_2$ and solve problem (\ref{t-DC_OP_The}) to obtain ${\hat{\bm{\Theta}}^i}$ and $\bm{\beta}^i$.
			\lIf{the objective value converges}{
				\textbf{break}
			}
		}
		Obtain $\hat{\bm{\theta}}^{t+1}$ via Cholesky decomposition.

		Sort ${\bm \beta}^t$ in ascending order as $\tilde{\bm \beta}^t$, where $\tilde{\bm \beta}^t_{\pi_ 1} \leq \tilde{\bm \beta}^t_{\pi_ 2}\leq \dots \leq \tilde{\bm \beta}^t_{\pi_ L}$.\\
		Initialize $J_{low}=0,J_{up}=L,J_0=0,flag=0$.\\
		\While{$J_{up}-J_{low}> 1$}
		{
			Set $J_0=\lfloor\frac{J_{up}+J_{low}}{2}\rfloor$. Let $|\bm{\hat{\theta}}^t_{N_{\text{total}+1}}| \!=\! 1, |\bm{\hat{\theta}}^t_i| \!=\! \tilde{\bm \beta}^t_{\pi_ l}  \rho_{\pi_ l},\! \forall \, \pi_ l \!\in\! \mathcal{L}, \!\forall \, i \!\in\! \mathcal{I}_{\pi_ l}$.\\
			Set $\tilde{\bm \beta}^t_{j}=0,\forall j\in \{1,\dots,J_0\}$, $\tilde{\bm \beta}^t_{j}=1,\forall j\in \{J_0+1,\dots,L\}$. \\
			Check the feasibility of problem (\ref{OP_act}).\\
			\eIf{problem (\ref{OP_act}) is feasible}
			{
				Recover $\bm{\tilde{\theta}}^t\!=\![\bm{\hat{\theta}}^t\!/\!\bm{\hat{\theta}}^t_{{{N}_{\text{total}}}+1}]_{(1:{{N}_{\text{total}}})}$, split $\bm{\tilde{\theta}}^t$ to get $\{ \bm{\tilde{\theta}}^t_1,\!\dots\!, \bm{\tilde{\theta}}^t_L\}$,
				$\mathcal{A}^t\!=\!\{ \pi_ l | \tilde{\bm \beta}^t_{\pi_ l}=1,\ \forall \, \pi_ l \in\! \mathcal{L}\}$, and $\bm \Theta^t_{l} = \mathrm{diag}\left( (\bm{\tilde{\theta}}^t_{ l})^{\sf{H}}\right), \forall \, l \in \mathcal{A}^t$.\\
				Set $J_{low}=J_0,flag=1$.
			}{Set $J_{up}=J_0$.} 
		}
		\lIf{The decrease of the network power consumption is below $\epsilon$ or flag=0}{
			\textbf{break}
		}
	}
	\KwOut{ $\mathcal{A}^t, \{ \bm{w}_k ^t\}, \{ \bm \Theta_l ^t\}$ }
	
\end{algorithm}		

	We present the proposed algorithm for solving the network power minimization problem (\ref{OP_Original}) in Algorithm~\ref{Alg}.
	To be specific, given $\{\bm \Theta^{t-1}_l\}$, problem (\ref{OP_Omega}) is a convex SOCP problem and thus can
	be efficiently solved via CVX~\cite{cvx} (line 2).
	Besides, given the results of problem (\ref{OP_Omega}), i.e., $\{\bm{w}_k^t\}$, problem (\ref{DC_OP_The}) is solved in two steps. 
	At step 1 (line 3-7), we obtain $\bm{\hat{\Theta}}^{*}$ via the proposed DC programming method. 
	At step 2 (line 8-23), we use a binary search to obtain the maximum number of inactive RIS.
	Specially, we first sort the ordering criterion ${\bm \beta}^t$ in an ascending order (line 8).
	Note that the relaxed ${\bm \beta}^t$ indicates the importance of RISs.
	Therefore, we iteratively add the RIS based on the order to determine the active RIS set (line10-20).
	The lower bound $J_{low}$ (resp. the upper bound $J_{up}$) is initially set as 0 (resp. $L$) (line 9). 
	At each round, we update $J_0$ as the average of $J_{up}$ and $J_{low}$ (line 11), and set the first $J_0$ RISs to be inactive while remaining others active (line 12). 
	Then, we check the feasibility of problem (\ref{OP_act}) (line 13). 
	If it is feasible, we recover $\bm \Theta_l^t$ based on the current settings of $\beta$ 
	(line 15) as well as $\bm{\tilde{\theta}}^*=[\bm{\hat{\theta}}^*/\bm{\hat{\theta}}_{({{N}_{\text{total}}}+1)}^*]_{(1:{{N}_{\text{total}}})}$ based on 
	$ \bm{\hat{\theta}}_{({{N}_{\text{total}}}+1)}^*$, where $[\bm{\hat{\theta}}^*/\bm{\hat{\theta}}_{({{N}_{\text{total}}}+1)}^*]_{(1:{{N}_{\text{total}}})}$ denotes 
	the vector that contains the first $N_{\text{total}}$ elements in $\bm{\hat{\theta}}^*/\bm{\hat{\theta}}_{({{N}_{\text{total}}}+1)}^*$. 
	The auxiliary binary variable set, denoted 
	as ${\bm \beta}^*$, is used to recover the active RIS set $\mathcal{A}^*=\{ l \,| \,{\beta}_{l}^*=1, \forall \, l \in \mathcal{L}\}$.
	By splitting $\bm{\tilde{\theta}}^*$, we obtain the set $\{ \bm{\tilde{\theta}}_1^*,\ldots,\bm{\tilde{\theta}}_L^*\}$, followed by $\bm \Theta_{l}^* = \mathrm{diag}((\bm{\tilde{\theta}}_{l}^*)^{\sf{H}}), \forall \, l \in \mathcal{L}$. 
	Otherwise, we update the upper bound as $J_0$ which means more active RISs are needed (line 18).
	
	
	Note that the proposed DC algorithm that jointly optimizes the active RIS set and phase shifts requires at most $\log(L)$ times to find the active RIS set, which is significantly lower than the exhaustive search approach that requires searching through $2^L$ possible combinations.

	\subsection{Low-Complexity Zero-Forcing (ZF) Method}\label{sec:ZF}
	In this subsection, we propose a low-complexity algorithm to minimize the network power consumption for the implemented hardware with worse computation capacity.
	
	We employ the ZF method to create orthogonal effective channels between the transmitter and receivers \cite{emil2014spm}.
	Since the ZF method for beamforming optimization at the BS has been widely investigated \cite{emil2014spm}, we devote to design a RIS shift optimization algorithm based on ZF in this subsection.
	We employ an exhausted search method to determine the active RIS set $\mathcal{A}$.
	
	Inspired by \cite{yang2021tc},  we first divide the elements of each RIS by $K$ parts, where each part of RIS elements only serves one user.
		Let $\tilde{\bm{\theta}}_k= [\bm{\tilde{\theta}}_{k,1}, \cdots, \bm{\tilde{\theta}}_{k,N_l}]^T$ represent the phase shifts of user $k$ and they can be adjusted by the RISs in $\mathcal{A}$.
		Let $\mathcal{I}_k$ represent the set of RIS elements adjusted to serve user $k$.
		Then, based on the ZF principle, we need to create orthogonal channels between the RISs and users.
		Recall the definition in (15), we denote $\bm{a}_{k,k}$ as $\bm a_{k,k}(l)=\mathrm {diag}(\bm h_{l,k}^{\sf{H}})\bm T_l \bm{w}_k$, $\forall \, k, \in \mathcal{K},\ \forall \, l \in \mathcal{L}$.
		It represents the effective channel between the BS and user $k$  via RIS $l$ for transmitting the intended signal of user $k$, and $\bm a_{k,j}(l)=\mathrm {diag}(\bm h_{l,k}^{\sf{H}})\bm T_l \bm{w}_j$ denotes the effective channel between the BS and user $k$  via RIS $l$ for transmitting the interfering signal of user $j$.
		Thus, the ZF phase shifts optimization problem of the active RISs can be formulated as:
	\begin{subequations}\label{opt:ZF}
		\begin{align}
			\mathop{\text{maximize}}_{\tilde{\bm{\theta}}_k}~~~ &|\tilde{\bm{\theta}}_k^H\bm{a}_{k,k}|\\
			\mathop{\text{subject to}}~~~&\tilde{\bm{\theta}}_k^H\bm{a}_{k,j}=0, \forall j\neq k,\label{equ:ZF_cons}\\
			&|\tilde{\theta}_{k,i}|=1, \forall i\in\mathcal{I}_k.\label{equ:unit_cons_ZF}
		\end{align}
	\end{subequations}
	Due to the non-convex unit-modulus constraints (\ref{opt:ZF}b), the conventional ZF method cannot be directly applied. With new non-convex constraints (\ref{opt:ZF}b), it is of importance to investigate the feasibility of (\ref{opt:ZF}). For the feasibility of problem (\ref{opt:ZF}), we have the following lemma.
	\begin{lemma}
		Problem (33) is feasible only if the following constraints are satisfied:
		\begin{align}\label{equ:lemma1}
			2\max_{i\in\mathcal{I}_k}|[\bm{a}_{k,j}]_i|\leq \sum_{i\in\mathcal{I}_k}|[\bm{a}_{k,j}]_i|, \forall j\neq k.
		\end{align}
		where $[\cdot]_i$ denotes the received signal of RIS phase shift $i$.
	\end{lemma}
	\begin{proof}
		Please refer to Appendix \ref{appd:lemma1}.
	\end{proof}
	If the received signals from RISs do not satisfy the condition in Lemma 1, then we need to solve the user admission control problem, which will be elaborated in Section IV.
	After checking the feasibility, finding the solution to the phase shifts optimization problem (\ref{opt:ZF}) is still challenging due to the non-convex unit modulus constraint (\ref{equ:unit_cons_ZF}).
	Without loss of generality, the term $|\tilde{\bm{\theta}}_k^H\bm{a}_{k,k}|$  in the objective function (\ref{opt:ZF})  can be expressed as a real number through an arbitrary phase rotation to phase shift $\tilde{\bm{\theta}}_k$ \cite{ami2006tsp}.
	Thus, we can convert the objective function (\ref{opt:ZF}) to
	\begin{align}
		\mathop{\text{maximize}}_{\tilde{\bm{\theta}}_k}~~~ &\Re(\tilde{\bm{\theta}}_k^H\bm{a}_{k,k}). \label{equ:opt1_ZF}
	\end{align}
	According to constraint (\ref{equ:ZF_cons}), $\tilde{\bm{\theta}}_k$ must lie in the orthogonal complement of the subspace span $\{\bm{a}_{k,j}, \forall j\neq k\}$ and the orthogonal projector matrix on this orthogonal complement is \cite{yang2021tc}
	\begin{align}
		\bm{Q}_k = \bm{I}-\bm{A}_k(\bm{A}_k^H\bm{A}_k)^{\dagger}\bm{A}_k^H,
	\end{align}
	where $\bm{A}_k=[\bm{a}_{k,1},\cdots,\bm{a}_{k,k-1},\bm{a}_{k,k+1},\cdots, \bm{a}_{K,k-1}]$.
	Therefore, $\tilde{\bm{\theta}}_k$ satisfying constraint (\ref{equ:ZF_cons}) can be expressed by 
	\begin{align}
		\tilde{\bm{\theta}}_k=\bm{Q}_k\bm{d}_k,\label{equ:theta_ZF}
	\end{align}
	where $\bm{d}_k$ is a complex vector. 
	Based on (\ref{equ:opt1_ZF}) and (\ref{equ:theta_ZF}), the ZF phase shifts optimization problem is equivalent to
	\begin{subequations}\label{equ:opt2_ZF}
		\begin{align}
			\mathop{\text{maximize}}_{\tilde{\bm{\theta}}_k, \bm{d}_k}~~~ &\Re (\tilde{\bm{\theta}}_k^H\bm{a}_{k,k})\\
			\mathop{\text{subject to}}~~~&\tilde{\bm{\theta}}_k=\bm{Q}_k\bm{d}_k,\label{equ:cons2_ZF}\\
			&|\tilde{\theta}_{k,i}|=1, \forall i\in\mathcal{I}_k.
		\end{align}
	\end{subequations}
	Due to the coupled variables $\tilde{\bm{\theta}}_k$ and $\bm{d}_k$ in constraint (\ref{equ:cons2_ZF}), we employ the barrier method to transform problem (\ref{equ:opt2_ZF}) to
	\begin{subequations}\label{equ:opt3_ZF}
		\begin{align}
			\mathop{\text{maximize}}_{\tilde{\bm{\theta}}_k, \bm{d}_k}~~~ &\Re (\tilde{\bm{\theta}}_k^H\bm{a}_{k,k}) - \delta||\tilde{\bm{\theta}}_k-\bm{Q}_k\bm{d}_k||^2\\
			\mathop{\text{subject to}}~~~
			&|\tilde{\theta}_{k,i}|=1, \forall i\in\mathcal{I}_k.
		\end{align}
	\end{subequations}	
	where $\delta$ represents the penalty factor.
	To solve the problem (\ref{equ:opt3_ZF}), we adopt an iterative algorithm to alternatively optimize $\tilde{\bm{\theta}}_k$ with fixed $\bm{d}_k$, and then update $\bm{d}_k$ with optimized $\tilde{\bm{\theta}}_k$ in the previous step, yielding an efficient closed-form solution in each step.

	\subsection{Computational Complexity and Convergence Analysis}\label{sec:analysis_npm}
	In this subsection, we analyze the proposed DC algorithm and ZF-based algorithm for the network power minimization problem.
	
	\textbf{Computational Complexity: }
	The minimization of the network power consumption contains two steps: 1) optimizing the beamforming vector $\{\bm{w}_k\}$ of BS in (\ref{OP_Omega}); 2) jointly optimizing the selected RIS set $\bm{\beta}$ and phase shifts of RIS $\{\hat{\bm{\Theta}}_l\}$ in (\ref{OP_Theta}).
	The optimization of the beamforming vector $\{\bm{w}_k\}$ is a convex SOCP problem with a computational complexity of $\mathcal{O}((MK)^{3.5})$ by using interior-point methods \cite{boyd2004convex}.
	Based on the proposed DC algorithm, we derive the solution to $\bm{\beta}$ and $\{\hat{\bm{\Theta}}_l\}$ by solving a sequence of the DC program $\mathscr{P}_{\text{DC1}}$.
	For each DC program, we need to solve the SDP problem (\ref{t-DC_OP_The}) at the $t$-th iteration.
	By employing the second-order interior point method \cite{luo2010sdr}, the computational complexity of each SDP problem is $\mathcal{O}((N_{\text{total}}^2+K)^{3.5})$. 
	We assume that the proposed DC program converges to critical points with $T > 1$ iterations, thus the computational complexity of problem $\mathscr{P}_{\text{DC1}}$ is $\mathcal{O}(T(N_{\text{total}}^2+K)^{3.5})$.
	
	The computational complexity of ZF mainly lies in the optimization of the channel matrix $\bm{Q}_k$ and exhausted search over all possible combinations of the active RIS sets, which is $\mathcal{O}(T(2^LN_{\text{total}}^3))$. 
	The number of RIS elements $N_{\text{total}}$ can be very large in practical scenarios, especially for large-scale wireless communication systems. 
		However, the number of RISs $L$ is typically much smaller than $N_{\text{total}}$ due to various constraints such as cost, power consumption, and physical space.
		As a result, the proposed ZF algorithm for RIS-empowered wireless communication has a shorter running time than the DC algorithm in practical implementation. 
	 On the other hand, the proposed DC algorithm has a higher computational complexity than the ZF algorithm, in order to obtain a high-quality solution.
	
	\textbf{Convergence Analysis:}
	Based on \cite{tao1997convex}, we provide the convergence property of the proposed DC algorithm for problem $\mathscr{P}_{\text{DC1}}$, which is given in Proposition 1.
	\begin{proposition}\label{prop:convergence}

		Let $f_1=\sum_{l  \in \mathcal{L}} {\beta}_l P_{\mathrm{RIS}}(N_l)+\rho\big(\text{Tr}(\bm{\hat{\Theta}})-\|\bm{\hat{\Theta}}\|_2\big)+\Delta_{\mathcal{C}}(\bm{\hat{\Theta}})$.
		The sequence $\{(\bm{\hat{\Theta}}^t)\}$ by iteratively solving problem (\ref{t-DC_OP_The}) for problem $\mathscr{P}_{\text{DC1}}$ has the following properties:
		\begin{enumerate}[1)]
			\item \label{property1} With any initial point, the sequence $\{(\bm{\hat{\Theta}}^t)\}$ converges to a critical point of $f_1$, where $\{f_1^t\}$ is strictly decreasing and convergent.
			\item \label{property2} For any $t=0,1,\cdots$, we have
			\begin{align}
				\text{Avg}\Big(\|\bm{\hat{\Theta}}^{t}-\bm{\hat{\Theta}}^{t+1}\|_F^2\Big)&\leq \frac{f_1^{0}-f_1^{\star}}{\alpha(t+1)},
			\end{align}
			where $f_1^{\star}$ is the global minimum of $f_1$ and $\text{Avg}\Big(\|\bm{\hat{\Theta}}^{t}-\bm{\hat{\Theta}}^{t+1}\|_F^2\Big)$ represents the average of the sequence $\{\|\bm{\hat{\Theta}}^{i}-\bm{\hat{\Theta}}^{i+1}\|_F^2\}_{i=0}^{t}$.
		\end{enumerate}   
	\end{proposition}
	\begin{proof}
		Please refer to Appendix \ref{appd:convergence} for details.
	\end{proof}
	Note that if there is a unique solution to each of the coordinate-wise minimization subproblems, then any limit point of the sequence of solutions is a stationary point \cite{bertsekas1997nonlinear}.
	Besides, the suboptimal AO algorithms have been shown to achieve almost the same performance as the globally optimal algorithm in \cite{bho2022arxiv}.
	Such properties reveal the efficiency of our proposed AO algorithm.
	
	\section{Alternating Framework for User Admission Control Problem}\label{Sec_UAC}
	The AO framework proposed in Section \ref{Sec_NPM} provides an efficient way to minimize the network power consumption.
	As the network state deteriorates, e.g., higher data rates of devices or worse channel conditions, the infeasibility issue may occur.
	As the network power minimization problem becomes infeasible, user admission control is required to enable the cellular network \cite{evaggelia2008ac, shi2016ac}. 
	In this section, we specify the algorithm proposed to solve (\ref{user_admission_control}).
	
	For notational ease, we define $\bm{{h}}_{k}^{\sf{H}} (\mathcal{L},\{ \bm \Theta_l \})\in \mathbb{C}^{1 \times M}$ as 
	\begin{eqnarray} \label{Eq_CompH_L}
		\bm{{h}}_{k}^{\sf{H}} (\mathcal{L},\{ \bm \Theta_l \}) = \sum_{l \in \mathcal{L}} \bm{h}_{l,k}^{\sf{H}} \bm{\Theta}_{l} \bm{T}_{l} + \bm h_{d,k}^{\sf{H}},\ \forall\  k \in \mathcal{S}. 
	\end{eqnarray}
	To facilitate the problem transformation, we rewrite the SINR constraint (\ref{P_SINR_s}) as follows
	\begin{eqnarray}\label{SINR_udc_c2}
		\gamma_k \left({\sum_{j \neq k} \left|\bm{{h}}_{k}^{\sf{H}}\big(\mathcal{L},\{ \bm \Theta_l \}\big) \bm{w}_j \right|^2 + \sigma_k^2} \right)\leq \nonumber\\
		{\left|\bm{{h}}_{k}^{\sf{H}}\big(\mathcal{L},\{ \bm \Theta_l \}\big) \bm{w}_k \right|^2}, \forall \, k \in \mathcal{S}.
	\end{eqnarray}
	To derive a feasible solution to problem (\ref{user_admission_control}), we add auxiliary variables $\{v_k\},v_k\in \mathbb{R}^+,k\in \mathcal{K}$ to the SINR constraint (\ref{user_admission_control}b) as
	\begin{align}
		\mathcal{C}_4({\mathcal{L},\{\bm{w}_k\},\{\bm{\Theta}_l\}},\{v_k\})\hspace{-1.2mm}: \hspace{-1.2mm}&\gamma_k^{\textrm{th}}\hspace{-1.2mm} \left({\sum_{j \neq k} \left|\bm{{h}}_{k}^{\sf{H}}\big(\mathcal{L},\{ \bm \Theta_l \}\big) \bm{w}_j \right|^2+ \sigma_k^2} \right)\hspace{-1.2mm}\leq \nonumber\\
		&{\left|\bm{{h}}_{k}^{\sf{H}}\big(\mathcal{L},\{ \bm \Theta_l \}\big) \bm{w}_k \right|^2}+v_k, \forall \, k \in \mathcal{K},
	\end{align}
	where $v_k$ represents the gap of received signal power for user $k$ to satisfy its desired SINR requirement \cite{zhao2015uac}. 
	With the help of auxiliary variables $\{v_k\}$, we could obtain the solution of problem (9) by firstly solving the following optimization problem 
	\begin{subequations} \label{udc_1}
		\begin{align}
			\mathop{\text{minimize}}_{ \{ \bm{w}_k \}, \{ \bm \Theta_l \}, \{v_k\}} \label{udc_c1}\hspace{3mm}  & P(\mathcal{L}, \{ \bm{w}_k \}) + \delta\sum_{k\in \mathcal{K}}v_k \\
			\text{subject to}\ \ 
			\hspace{3mm}  & \mathcal{C}_4({\mathcal{L},\{\bm{w}_k\},\{\bm{\Theta}_l\}},\{v_k\}),(\ref{OP_Original}c), (\ref{user_admission_control}d),
		\end{align}
	\end{subequations}
	where $\delta$ is a positive large constant variable to guarantee the optimization problem (\ref{udc_1}) always be feasible \cite{evaggelia2008ac}.
	To solve the problem (\ref{udc_1}), we follow Section III and propose an AO framework to alternatively optimize the beamforming vector and the selected user set, and the phase shifts of RISs.	

	\subsection{Joint Optimization for Beamforming  and User Admission}
	In this subsection, for given phase shifts of RISs, we optimize the transmit beamforming vectors $ \{ \bm{w}_k \}$ at the BS and the auxiliary variables $\{v_k\}$ for given RIS phase-shift matrices. The problem is formulated as
	\begin{subequations} \label{udc_subpro1}
		\begin{align}
			\mathop{\text{minimize}}_{ \{ \bm{w}_k \}, \{v_k\}} \label{udc_subpro1_c1}\hspace{3mm} \ \  & \sum_{k \in \mathcal{K}} \frac{1}{\eta} \| \bm{w}_k \|^2 + \sum_{l \in \mathcal{L}}P_{\mathrm{RIS}}(N_l) + \delta\sum_{k\in \mathcal{K}}v_k \\
			\text{subject to}\ \ 
			\label{udc_subpro1_c2} \hspace{3mm}  & \mathcal{C}_4(\{\bm{w}_k\},\{v_k\}), (\ref{OP_Original}c).
		\end{align}
	\end{subequations}
	We define $\bm{W}_k=\bm{w}_k\bm{w}_k^{\sf H}\in \mathbb{C}^{M\times M}$ with $\bm{W}_k\succcurlyeq 0$ and $\text{rank}(\bm{W}_k)=1$. 
	We define $\bm{H}_k=\bm{{h}}_{k}\bm{{h}}_{k}^{\sf{H}}\in \mathbb{C}^{M\times M},\ \forall k\in \mathcal{K}$, and constraint (\ref{udc_subpro1_c2}) can be rewritten as
	\begin{eqnarray}
		\mathcal{C}_5(\{\bm{W}_k\},\{v_k\}): 	\gamma_k^{\textrm{th}} \left({\sum_{j \neq k} \text{Tr}\Big(\bm{H}_k\bm{W}_j\Big)+ \sigma_k^2} \right)\leq \nonumber\\ \text{Tr}\Big(\bm{H}_k\bm{W}_k\Big)+v_k, \forall \, k \in \mathcal{K}.
	\end{eqnarray}
	As $ \| \bm{w}_k \|^2=\text{Tr}(\bm{W}_k),\ \forall k\in \mathcal{K}$, problem (\ref{udc_subpro1}) is further reformulated as
	\begin{subequations} \label{udc_subpro2}
		\begin{align}
			\mathop{\text{minimize}}_{ \{ \bm{W}_k \},\{v_k\}} \label{udc_subpro2_c1}\hspace{3mm}\ \   & \sum_{k \in \mathcal{K}} \frac{1}{\eta} \text{Tr}(\bm{W}_k) + \sum_{l \in \mathcal{L}}P_{\mathrm{RIS}}(N_l) + \delta\sum_{k\in \mathcal{K}}v_k \\
			\text{subject to}\ \ 
			\label{udc_subpro2_c2} \hspace{3mm}  & \mathcal{C}_5(\{\bm{W}_k\},\{v_k\}),\\
			&\sum_{k\in \mathcal{K}}  \text{Tr}(\bm{W}_k) \le P_{\mathrm{max}},\\
			&\bm W_k \succeq 0, \forall k\in \mathcal{K},\\
			&\text{rank}(\bm{W}_k)=1,\forall k\in \mathcal{K}.
		\end{align}
	\end{subequations}
	Due to the rank-one constraint in (\ref{udc_subpro2}e).
	In Section \ref{Sec_UAC}-C, following Section III, we will specify the DC programming method proposed to solve (\ref{udc_subpro2}).

	\subsection{Phase Shifts Optimization}
	Given beamforming vectors $\{\bm{w}_k\}$ and auxiliary variables $\{v_k\}$, the resulting optimization problem for the phase shifts of RISs can be converted as
	\begin{subequations} \label{udc_RIS}
		\begin{align}
			\mathop{\text{find}} \label{udc_RIS_c1}\hspace{3mm} \ \  &  \{\bm{\Theta}_l\}\\
			\text{subject to}\ \ 
			\label{udc_RIS_c2} \hspace{3mm}  & \mathcal{C}_4(\{\bm{\Theta}_l\}),  (\ref{user_admission_control}d)
		\end{align}
	\end{subequations}	
	Note that ``find" in problem (\ref{udc_RIS}) is used only to check the feasibility of the problem. 
		When we can ``find'' a feasible solution to problem (\ref{udc_RIS}), the AO algorithm stops and the current number of admitted users is the largest number of users that the cellular network can serve.

	Let $\bm{\theta}_l=[\theta_{l,1},\theta_{l,2},\ldots,\theta_{l,N_l}]^{\sf H} \in \mathbb{C}^{N_l \times 1},\ \forall \, l \in \mathcal{L}$.
	Similar to Section \ref{sec:3_b}, we have $(\bm h_{l,k}^{\sf{H}} \bm \Theta_l \bm T_l +\bm h_{d,k}^{\sf{H}})\bm{w}_j = \bm \theta_l^{\sf{H}}\bm a_{k,j}(l)+b_{k,j} , \forall \, k, j \in \mathcal{K}, \forall \, l \in \mathcal{L}$. 
	Then the SINR constraint (\ref{udc_RIS_c2}) can be given as 
	\begin{eqnarray} \label{SINR2_udc_c2}
		\hspace{-16mm} \gamma_k^{\textrm{th}} \left (\sum_{j \neq k} \Big| \sum_{l \in \mathcal{L}} \bm \theta_l^{\sf{H}}\bm a_{k,j}(l) + b_{k,j} \Big|^2 + \sigma_k^2 \right)  \leq \nonumber\\
		 \left| \sum_{l \in \mathcal{L}} \bm \theta_l^{\sf{H}}\bm a_{k,k}(l) + b_{k,k} \right|^2+v_k, \forall \, k\in \mathcal{K}.
	\end{eqnarray}
	
	By defining $\bm{\bar{\theta}}=[\bm{\theta}_1^{\sf{H}},\bm{\theta}_2^{\sf{H}},\dots,\bm{\theta}_L^{\sf{H}}]^{\sf{H}} \in \mathbb{C}^{{{N}_{\text{total}}}\times 1}$ and $\bm{\check{\theta}}=\left[\begin{array}{l}\bm{\bar{\theta}} \\ {1}
		\end{array}
		\right]\in \mathbb{C}^{({{N}_{\text{total}}}+1) \times 1}$,
	the non-convex constraint (\ref{SINR2_udc_c2}) is rewritten as
	\begin{eqnarray}
		\mathcal{C}_6(\mathcal{L},\bm{\check{\theta}}): \hspace{-6mm}  \gamma_k^{\textrm{th}} \left(\sum_{j \neq k} \left( \bm{\check{\theta}}^{\sf{H}}\bm R_{k,j} \bm{\check{\theta}} + |b_{k,j}|^2\right)  + \sigma_k^2 \right) \leq \nonumber\\ \bm{\check{\theta}}^{\sf{H}}\bm R_{k,k} \bm{\check{\theta}} + |b_{k,k}|^2+v_k, \forall \, k \in \mathcal{K}.
	\end{eqnarray}
	Then problem (\ref{udc_RIS}) is reformulated as
	\begin{subequations} \label{udc_RIS2}
		\begin{align}
			\mathop{\text{find}} \label{udc_RIS2_c1}\hspace{3mm} \ \  & \bm{\check{\theta}} \\
			\text{subject to}\ \ 
			\label{udc_RIS2_c2} \hspace{3mm}  & \mathcal{C}_6(\mathcal{L},\bm{\check{\theta}}),  | \bm{\check{\theta}}_{n}| = 1,\ \forall\ n\in \mathcal{N}_{\text{total}}^*,
		\end{align}
	\end{subequations}
	where the index set is defined as $\mathcal{N}_{\text{total}}^*=\{1,2,\cdots,N_{\text{total}}+1\}$.
	Next, we define a PSD matrix $\bm{\check{\Theta}}=\bm{\check{\theta}}\bm{\check{\theta}}^{\sf H}\in \mathbb{C}^{({{N}_{\text{total}}}+1) \times ({{N}_{\text{total}}}+1)}$ with $\bm{\check{\Theta}} \succeq 0$ and ${\rm {rank}}(\bm{\check{\Theta}})=1$.
	Hence, we have $\mathrm{Tr}(\bm R_{k,j} \bm{\check{\Theta}}) = {\bm{\check{\theta}}}^{\sf{H}}\bm R_{k,j} {\bm{\check{\theta}}}$.
	Then problem (\ref{udc_RIS2}) becomes
	\begin{subequations} \label{udc_RIS3}
		\begin{align}
			\mathop{\text{find}} \label{udc_RIS3_c1}\hspace{3mm} \ \  &\bm{\check{\Theta}} \\
			\text{subject to}\ \ 
			\label{udc_RIS3_c2} \hspace{3mm}  & \mathcal{C}_7(\mathcal{L},\bm{\check{\Theta}}), \\
			\label{udc_RIS3_c4} \hspace{3mm}  & \bm{\check{\Theta}} \succeq 0, | \bm{\check{\Theta}}_{n,n}| = 1,\ \forall\ n\in \mathcal{N}_{\text{total}}^*, \\
			&\text{rank}(\bm{\check{\Theta}})=1,
		\end{align}
	\end{subequations}
	where the constraint (\ref{udc_RIS3_c2}) is given by
	\begin{align}
		\mathcal{C}_7(\mathcal{L},\bm{\check{\Theta}}):  \gamma_k^{\textrm{th}} \left(\sum_{j \neq k} \left( \text{Tr}(\bm{R}_{k,j} \bm{\check{\Theta}}) + |b_{k,j}|^2\right)  + \sigma_k^2 \right) \leq  \nonumber\\
		\text{Tr}(\bm{R}_{k,k} \bm{\check{\Theta}}) + |b_{k,k}|^2+v_k, \forall \, k \in \mathcal{K}.
	\end{align}
	Note that both problems (\ref{udc_subpro2}) and (\ref{udc_RIS3}) are non-convex rank-constrained optimization problems, which are similar to problem (\ref{OP_The_rank_one}) in Section III. 
	Therefore, we adopt the same DC algorithm to solve problems (\ref{udc_subpro2}) and (\ref{udc_RIS3}).
	
	\subsection{Alternating DC Framework}

	Similarly, problems (\ref{udc_subpro2}) and (\ref{udc_RIS3}) can be reformulated by adding the DC
	representation of the rank-one constraint into the objective function.
	We can rewrite the beamforming optimization as

	\begin{subequations} \label{udc_subpro3}
		\begin{align}
			\mathscr{P}_{\text{DC2}}:\mathop{\text{minimize}}_{ \{ \bm{W}_k \}, \{v_k\}} \label{udc_subpro3_c1}
			\ \ &\sum_{k \in \mathcal{K}} \frac{1}{\eta} \text{Tr}(\bm{W}_k) +  \sum_{l \in \mathcal{L}}P_{\mathrm{RIS}}(N_l) +\nonumber\\ &\delta\sum_{k\in \mathcal{K}}v_k +\zeta\sum_{k\in \mathcal{K}}(\text{Tr}(\bm{W}_k)-\|\bm{W}_k\|_2)\\
			\text{subject to}\ \ 
			\label{udc_subpro3_c2}   & \mathcal{C}_5(\{\bm W_k\},\{v_k\}), (\ref{udc_subpro2}c), 
		\end{align}
	\end{subequations}
	where $\zeta$ represents the penalty parameter.
	
	The phase-shift matrices optimization problem for given $\{ \bm{w}_k \}$ and $ \{v_k\}$ is
	\begin{subequations} \label{udc_RIS4}
		\begin{align}
			\mathscr{P}_{\text{DC3}}:\mathop{\text{minimize}}_{\bm{\check{\Theta}}} \label{udc_RIS4_c1}\hspace{3mm} \ \  & \text{Tr}(\bm{\check{\Theta}})-\|\bm{\check{\Theta}}\|_2 \\
			\text{subject to}\ \ 
			\label{udc_RIS4_c2} \hspace{3mm}  & \mathcal{C}_7(\mathcal{L},\bm{\check{\Theta}}), (\ref{udc_RIS3}c).
		\end{align}
	\end{subequations}

	Thus, following Section \ref{sec:3_c}, the $t$-th iteration of $\{\bm{W}_k^t\}$ can be obtained by solving the following problem
	\begin{subequations} \label{udc_subpro4}
		\begin{align}
			\mathop{\text{minimize}}_{ \{ \bm{W}_k \}, \{v_k\}} \label{udc_subpro4_c1}
			\ \ &\sum_{k \in \mathcal{K}} \frac{1}{\eta} \text{Tr}(\bm{W}_k) +  \sum_{l \in \mathcal{L}}P_{\mathrm{RIS}}(N_l) +\nonumber\\
			&\delta\sum_{k\in \mathcal{K}}v_k+\zeta\sum_{k\in \mathcal{K}}(\langle \bm{W}_k ,\bm{I}_M-\partial_{\bm{W}_k^{t-1}}\|\bm{W}_k\|_2\rangle)\\
		\text{subject to}\ \ 
		\label{udc_subpro3_c2}   & \mathcal{C}_5(\{\bm W_k\},\{v_k\}), (\ref{udc_subpro2}c), (\ref{udc_subpro2}e),
		\end{align}
	\end{subequations}
	and the $t$-th iteration $\{\bm{\check{\Theta}}^t\}$ can be optimized by
	\begin{subequations} \label{udc_RIS5}
		\begin{align}
			\mathop{\text{minimize}}_{\bm{\check{\Theta}}} \label{udc_RIS5_c1}\hspace{3mm} \ \  & \text{Tr}(\bm{\check{\Theta}})-\langle \bm{\check{\Theta}},\partial_{\bm{\check{\Theta}}^{t-1}}\|\bm{\check{\Theta}}\|_2\rangle \\
			\text{subject to}\ \ 
			\label{udc_RIS4_c2} \hspace{3mm}  & \mathcal{C}_7(\mathcal{L},\bm{\check{\Theta}}), (\ref{udc_RIS3}c).
		\end{align}
	\end{subequations}
	By solving the convex optimization problem (\ref{udc_subpro4}) and (\ref{udc_RIS5}), we can find the solution to problem (\ref{udc_subpro3}) and (\ref{udc_RIS4}).
	The proposed algorithm for solving problems (\ref{udc_subpro3}) and (\ref{udc_RIS4}) is illustrated in Algorithm \ref{Alg2}, and then the overall algorithm for user admission control is summarized in Algorithm~\ref{Alg3}.
	\begin{algorithm}[t] 
		\caption{Proposed Algorithm for Solving Problem (53) and (54)}\label{Alg2}
		\KwIn{$\bm{\check{\Theta}}^0$ and $\epsilon>0$.}
		\For{$t=1,2,...$}{
			Given $\bm{\check{\Theta}}^0$, solve problem (53) to obtain $\{\bm{W}_k^t\}$.\\
			\For{$i=1,2,...$}{
				Select a subgradient of $\partial_{W_k^{i-1}}\|\bm{W}_k\|_2$ and solve problem (55) to obtain $\{\bm{W}_k^i\}$ and $v_k^i$.\\
				\lIf{the objective value converges}{
					\textbf{break}
				}
			}
			
			Given $\{\bm{W}_k^t\}$ and $v_k^t$, solve problem (54) for $\bm{\check{\Theta}}^{t+1}$.\\
			\For{$i=1,2,...$}{
				Select a subgradient of $\partial_{\bm{\check{\Theta}}^{i-1}}\|\bm{\check{\Theta}}\|_2$ and solve problem (56) to obtain ${\bm{\check{\Theta}}^i}$.\\
				\lIf{the objective value converges}{
					\textbf{break}
				}
			}
			Obtain $\bm{\check{\theta}}^{t+1}$ via Cholesky decomposition.
			
			\lIf{the decrease of the objective value of problem (53) is below $\epsilon$ or problem (54) becomes infeasible}{
				\textbf{break}
			}
		}
		
	\end{algorithm}	
	\begin{algorithm}[!ht] 
		\caption{Overall algorithm for user admission control}\label{Alg3}
		\KwIn{Initialize set $\mathcal{S}=\{1,...,K\}$.}
		Adopt Algorithm 2 to solve problem (\ref{udc_subpro3}) and (\ref{udc_RIS4}), and get the solution $\bm{w}_k$ and $\bm{\check{\Theta}}$ for each user $k$ in $\mathcal{S}$\;
		\For{each user $k\in\mathcal{S}$}{
			\If{the QoS constraint (\ref{SINR_udc_c2}) is not satisfied}{
				Remove user $k^*$ with the largest $v_k$. Go to step 1.\\

			}
		}
	\end{algorithm}	

	\subsection{Computational Complexity and Convergence Analysis}
	In this subsection, we analyze the proposed DC algorithm for the user admission control problem.
	Similar to Subsection \ref{sec:analysis_npm}, we demonstrate the theory performance of Algorithm 3 for the user admission control problem as follows.
	
	\textbf{Computational Complexity: }
	We solve the user admission control problem by two steps: 1) jointly optimizing the beamforming vector of BS and user admitted set in (\ref{udc_subpro1}); 2) optimizing the phase shifts of RISs in (\ref{udc_RIS}).
	Both the two problems are solved by the proposed DC algorithm, where the DC formulations are given in $\mathscr{P}_{\text{DC2}}$ and $\mathscr{P}_{\text{DC3}}$.
	The SDP problems of $\mathscr{P}_{\text{DC2}}$ and $\mathscr{P}_{\text{DC3}}$ are (\ref{udc_subpro4}) and (\ref{udc_RIS5}), and the respective computational complexities are $\mathcal{O}(T(M^2+K)^{3.5})$ and $\mathcal{O}(T(N_{\text{total}}^2+K)^{3.5})$.
	
	\textbf{Convergence Analysis:}
	Without loss of generality, the convergence of the proposed DC algorithm for problem $\mathscr{P}_{\text{DC2}}$ and $\mathscr{P}_{\text{DC3}}$ can be proved similarly to Proposition 1.
	
	\begin{figure*}[ht] \centering
		\includegraphics[scale = 0.32]{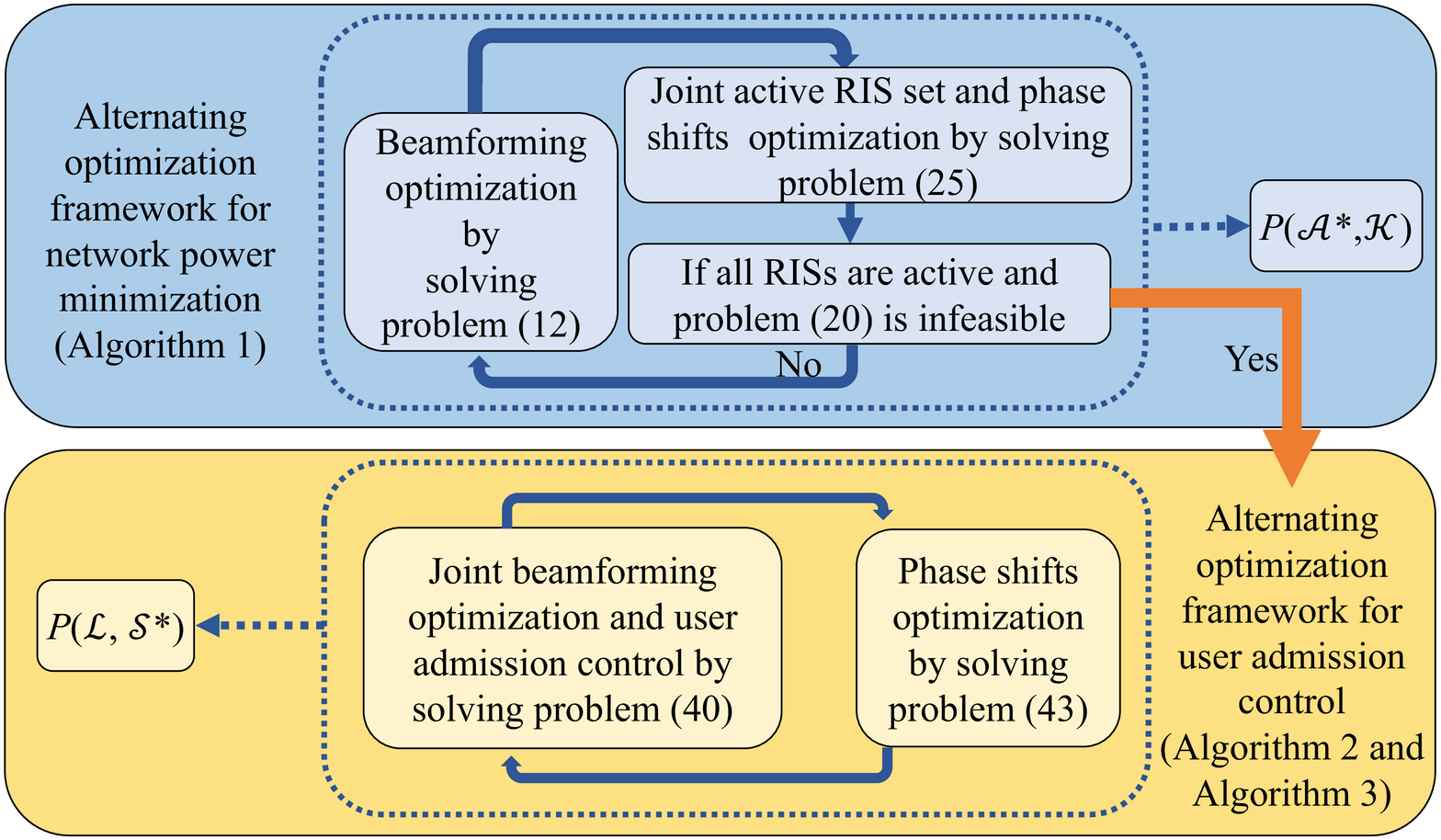}
		\caption{Optimization framework of network power minimization and user admission control.}
		\label{fig:framework}
	\end{figure*}
	\subsection{Unified Framework for Network Power Minimization and User Admission Control}

	Based on the above discussion, we illustrate the proposed unified AO framework for both network power minimization and user admission control in Fig.~\ref{fig:framework}. 
		We first attempt to minimize the network power by optimizing the phase shifts, the active set of RISs, and beamforming vectors at the BS (Algorithm 1).
		When the solution becomes infeasible, we convert to solve the user admission control problem, and employ the proposed AO algorithm to optimize the phase shifts of RISs, the beamforming vectors at the BS, and the active user set (Algorithm 3).
		Both of the two schemes proposed in Sections III and IV are aimed to minimize the network power consumption, i.e., achieving a lower power consumption, but in different network conditions.
		Therefore, the proposed schemes in Sections III and IV are complementary and the proposed unified optimization framework can efficiently minimize the power consumption by combining the two schemes, thereby improving the energy efficiency of green wireless networks.

	\section{Numerical Results}\label{Sec_NumRes}

	In this section, we demonstrate the effectiveness of the proposed algorithms in a multi-RIS-assisted MISO cellular network via numerical results. 

	We consider a wireless network with a uniform linear array (ULA) at the BS and a uniform rectangular array (URA) at each RIS. 
	The number of transmit antennas at the BS is set as $10$.
	The number of RISs is set as $L=3$.
	The location of the BS is set as $(0,0,10)$, and the locations of RISs are set as $(0,30,10)$, $(30,70,10)$, and $(70,0,10)$, respectively.
	The users are assumed to be randomly located in a circle centered at $(70,40,0)$ with radius of 15 meters. 
	The detailed locations of the BS, RISs and users are illustrated in Fig.~\ref{fig:x-y-model}.
	We follow the simulation setting as in \cite{huang2019reconfigurable,hua2021greenris} with the path loss parameter with $\alpha_{BU}=3.67$, $\alpha_{BR}=2.2$, and $\alpha_{RU}=2$ for BS-users link, the BS-RIS link, and RIS-user link, respectively \cite{3GPP}.

\begin{figure}
	\centering
	\begin{minipage}[t]{\linewidth}
		\includegraphics[scale = 0.28]{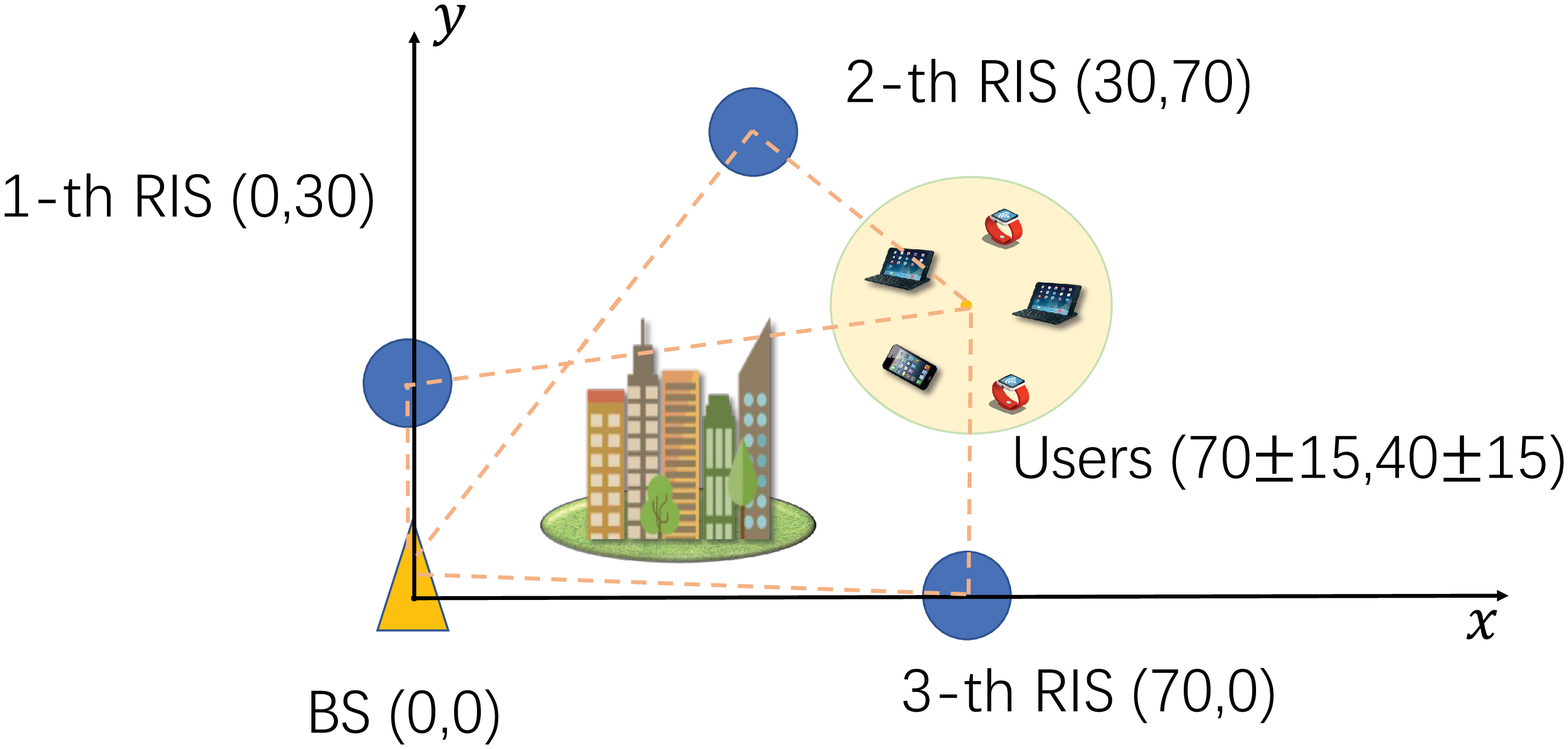}
		\caption{Horizontal location of the BS, RISs, and users.}
		\label{fig:x-y-model}
	\end{minipage}\\

	\begin{minipage}[t]{\linewidth}
		\centering
		\includegraphics[scale = 0.5]{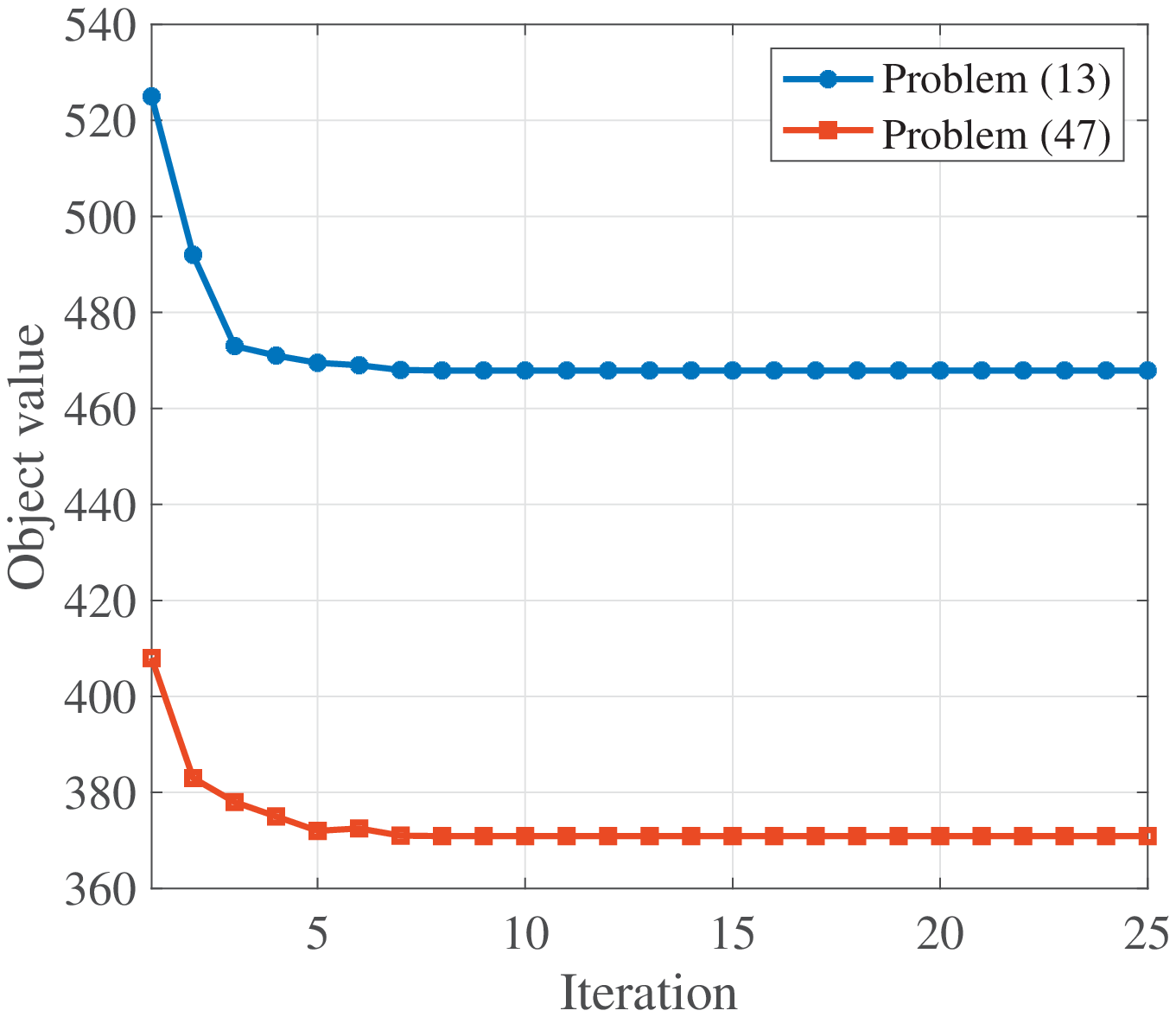}
		\caption{Convergence of the proposed DC algorithm.}
		\label{fig:npm_conv}
	\end{minipage}
\end{figure}
	
	In our simulations, we assume the SINR threshold of all users is the same, which is set as $1$ dB.
	Besides, we set $\sigma_k^2=-40$ dBm, $\eta=0.6$, $P_{\max}=1000$ milliwatt (mW) and $P_{\mathrm{RE}}=10$ mW \cite{huang2019reconfigurable}. 
	All results are averaged over 1000 Monte Carlo simulations.

	Firstly, we study the convergence performance of the proposed DC algorithm for the network power minimization and user admission control problem, respectively.
	In Fig.~\ref{fig:npm_conv}, we depict the objective value decrease (\ref{udc_subpro1}) as the number of iteration increases.
	We observe that both the objective values for the network power minimization problem and the user admission control problem decrease monotonically.
	In addition, the objective value decreases fast in the first several iterations, which demonstrates the convergence speed performance of the proposed DC algorithm.

	\subsection{Network Power Minimization}
	In this subsection, we validate the effectiveness of our proposed DC algorithm for the network power minimization problem.
	We fix the number of reflecting elements at each RIS to $N_l=20$ and the number of users to $K=6$.
	We consider the following three schemes as baselines:
	\begin{itemize}
		\item \textbf{SDR}: the BS leverages the SDR algorithm to solve the power minimization problem \cite{sun2020green}.
		\item \textbf{Successive Convex Approximation (SCA)}: the BS exploits the SCA algorithm to minimize the network power consumption \cite{yang2021distributed}.
		\item \textbf{ZF}: the BS employs the ZF algorithm proposed in Sec. \ref{sec:ZF}  to alternatively optimize the BS beamforming and RISs phase shifts with an exhausted search over all combinations of the active RISs sets.
		\item \textbf{Exhaustive search}: all possible combinations of active RISs are checked. Note that the complexity of the exhaustive search grows exponentially with the number of RISs.
		\item \textbf{All active}: all RISs in $\mathcal{L}$ are active and only the transmit power of the BS is minimized. 
	\end{itemize}
\begin{table}[!t]
	\centering
	\begin{minipage}{\linewidth}
		\caption{The average number of active RISs with different algorithms}
		\centering
		\label{table:avg_num} 
		\resizebox{0.9\linewidth}{!}{
		\begin{tabular}{c |c |c |c |c |c}  
			\hline
			& & & \\[-10pt]  
			Target SINR [dB]&0.5&1&1.5&2&2.5 \\ 
			\hline  
			& & & \\[-10pt]  
			All Active&3&3&3&3&3 \\  
			\hline
			& & & \\[-10pt]  
			Proposed DC&0.5&0.8&1&1.3&1.8 \\
			\hline
			& & & \\[-10pt]  
			SDR&0.5&0.9&1.1&1.5&2 \\
			\hline
			& & & \\[-10pt]  
			SCA&0.8&1.2&1.5&1.8&2.4 \\
			\hline 
			& & & \\[-10pt]  
			ZF&0.9&1.4&1.8&2.2&2.6 \\
			\hline
			& & & \\[-10pt]  
			Exhaustive Search&0.4&1&1.4&2&1.8 \\
			\hline
		\end{tabular}}
	\end{minipage}\\
\vspace{5mm}
	\begin{minipage}{\linewidth}
		\centering
		\caption{The total RIS power consumption with different algorithms (mW)}  
		\label{table:pow1}  
		\resizebox{\linewidth}{!}{
		\begin{tabular}{c| c| c| c| c| c}  
			\hline
			& & & \\[-10pt]  
			Target SINR [dB]&0.5&1&1.5&2&2.5 \\ 
			\hline  
			& & & \\[-10pt]  
			All Active &135&135&135&135&135 \\  
			\hline  
			& & & \\[-10pt]  
			Proposed DC&22.5&36&45&58.5&81 \\
			\hline  
			& & & \\[-10pt]  
			SDR&22.5&40.5&49.5&67.5&90 \\
			\hline
			& & & \\[-10pt]  
			SCA&36&54&67.5&81&108 \\
			\hline 
			& & & \\[-10pt]  
			ZF&40.5&63&81&99&117 \\
			\hline  
			& & & \\[-10pt]  
			Exhaustive Search&18&45&63&90&81 \\
			\hline
		\end{tabular}}
	\end{minipage} 
\end{table}
	
	We first study the impact of the target SINR on the proposed RIS-assisted green wireless network and the proposed algorithms. To this end, we vary the target SINR and show the average number of active RISs, the total RIS power consumption and the total transmit power consumption in Table \ref{table:avg_num}, Table \ref{table:pow1} and Table \ref{table:pow2}, respectively. 
	Note that some of the average numbers of active RIS in Table \ref{table:avg_num} are non-integers since we run the algorithms over 1000 Monte Carlo simulations, which can help to clearly show the performance of different algorithms.
	We have the following observations.
	
	First, as shown in Table \ref{table:pow2}, the scheme with all RISs being active achieves the maximum transmit power consumption at the BS since the network performance, e.g., QoS, can be improved by employing more RISs. 
	However, the active RISs incur large RIS power consumption, as shown in Table \ref{table:pow1}. Therefore, it suffers from the highest network power consumption, which is power inefficient.
	Second, the ZF, SCA, SDR, and the proposed DC algorithms significantly reduce the network power consumption compared to the scheme with all RISs being active.
	In particular, the proposed DC outperforms the ZF, SCA, and SDR algorithms, and it achieves a comparable result with the optimal exhaustive method.
	Therefore, we obtain the following conclusions:
	1) the proposed DC, ZF, SCA, and SDR algorithms reduce the RIS power consumption (as shown in Table \ref{table:pow1}) by switching off some RISs (as shown in Table \ref{table:avg_num}). 
	2) the proposed DC, ZF, SCA, and SDR schemes require more transmit power at the BS than the scheme with all RISs being active (as shown in Table \ref{table:pow2}). 
	This is mainly because with some RISs being switched off, more power is needed at the BS to guarantee the QoS requirements. 
	3) the SDR algorithm achieves higher performance than the SCA algorithm since it convexifies the non-convex quadratic constraints as a linear constraints via dropping the rank-one constraint in the lifting problem.
	The proposed DC outperforms SDR in terms of the RIS power consumption, the transmit power consumption and the network power consumption, due to its superiority in capturing the rank-one constraint.
	
	\begin{figure}
		\centering
		\hspace{-1mm}
		\begin{minipage}{\linewidth}
			\captionof{table}{The average total transmit power consumption with different algorithms (mW)}
			\label{table:pow2}  
			\begin{tabular}{c| c| c| c| c| c}  
				\hline  
				& & & \\[-10pt]  
				Target SINR [dB]&0.5&1&1.5&2&2.5 \\ 
				\hline  
				& & & \\[-10pt]  
				All Active&199&335&378&464&629 \\  
				\hline 
				& & & \\[-10pt]  
				Proposed DC&125&198&264&353&472 \\
				\hline 
				& & & \\[-10pt]  
				SDR&140&263&318&372&569 \\
				\hline
				& & & \\[-10pt]  
				SCA&145&283&352&405&603 \\
				\hline 
				& & & \\[-10pt]  
				ZF&162&294&365&420&611 \\
				\hline
				& & & \\[-10pt]  
				Exhaustive Search&121&193&244&325&437 \\
				\hline
			\end{tabular}
		\end{minipage}

	\end{figure}
	\begin{figure}
		\centering
		\begin{minipage}{\linewidth}
			\centering
			\includegraphics[scale = 0.4]{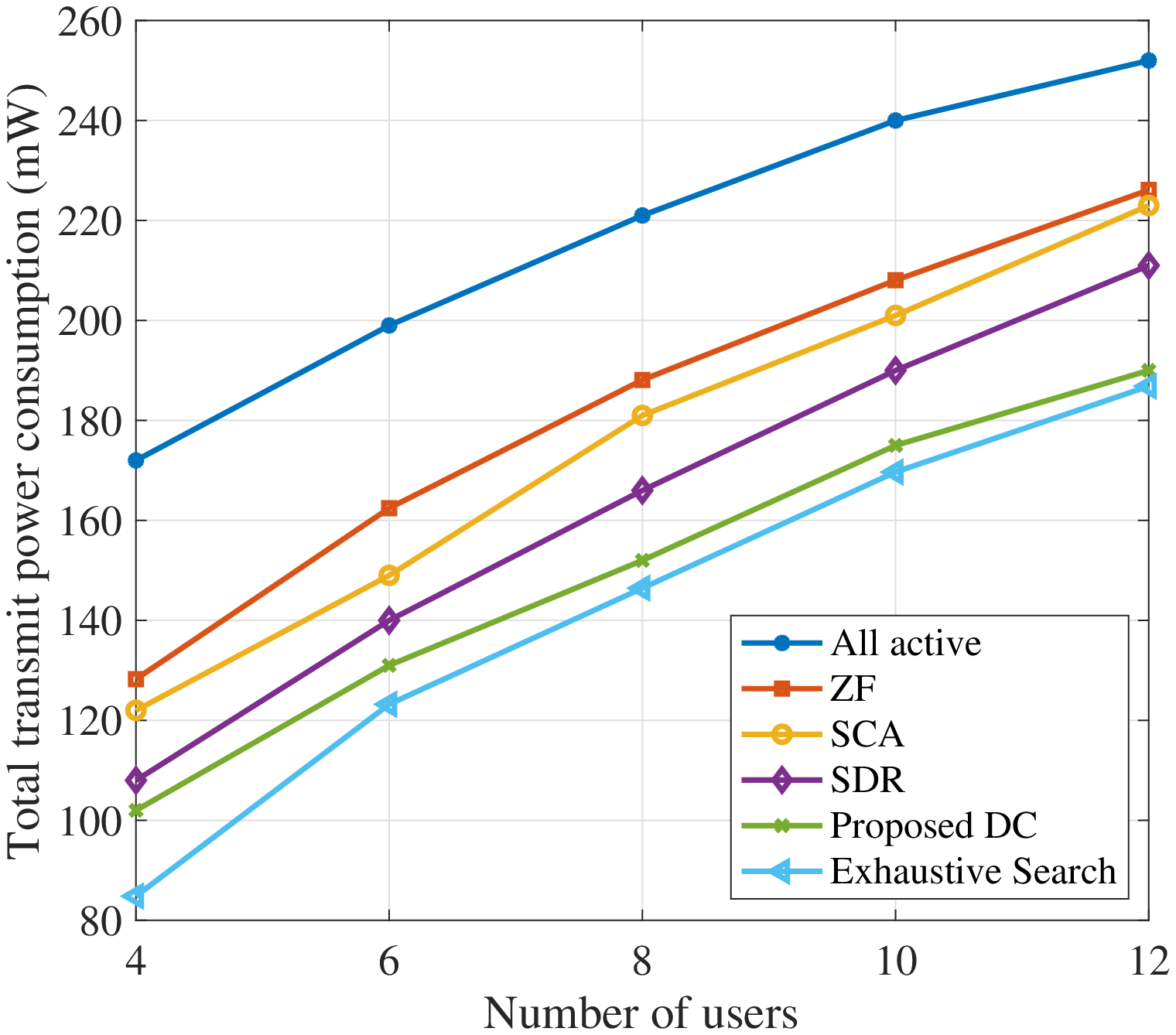}
			\caption{Total transmit power consumption versus different number of users $K$.}
			\label{fig:npm_diff_k}
		\end{minipage}
	\end{figure}
	
	Fig.~\ref{fig:npm_diff_k} shows the total transmit power consumption versus the number of RISs. 
	From Fig.~\ref{fig:npm_diff_k}, we can observe that the total transmit power consumption monotonically increases with the number of users. 
	A larger number of users can lead to higher power consumption due to the SINR constraints for each user.
	However, the proposed DC algorithm always outperforms the ZF, SCA and SDR algorithms.
	The reason mainly comes from the sufficient exploitation of matrix lifting technique and rank-one constraints, therefore achieving superb performance for network power minimization.
	Note that the ZF algorithm has much lower computational complexity than all other algorithms.
	Therefore, in practice, both the different two types of optimization algorithms (namely, ZF and DC) can be applied according to different requirements.
	For instance, the proposed ZF is more suitable for computation-sensitive scenarios while the proposed DC is more suitable for computation-insensitive scenarios to enhance the performance.
	\begin{figure}[t]
		\centering
		\includegraphics[scale = 0.4]{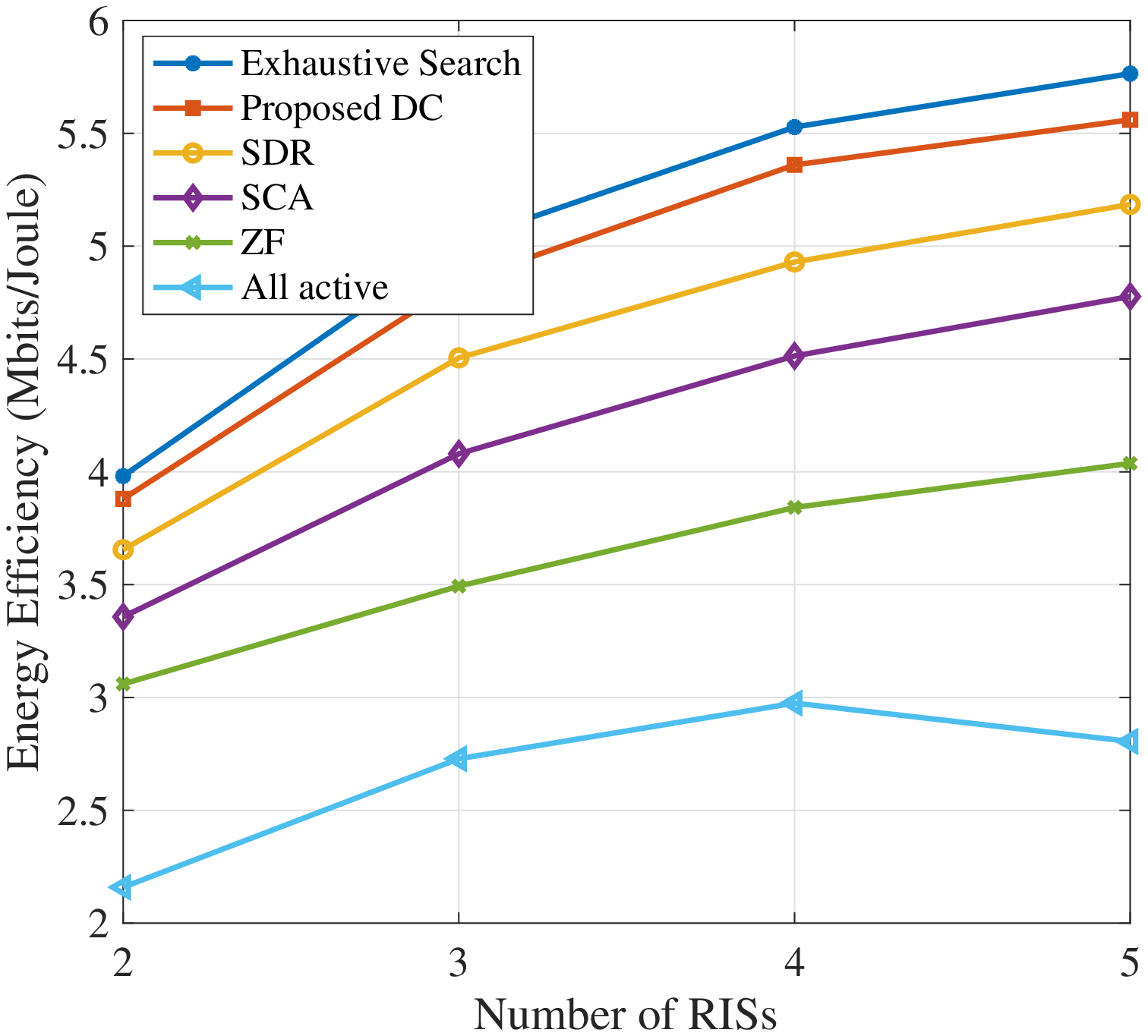}
		\caption{Energy efficiency versus different number of RISs $L$.}
		\label{fig:ee_diff_k}
	\end{figure}
	Moreover, we illustrate the energy efficiency performance of the baseline algorithms in Fig.~\ref{fig:ee_diff_k}.
	The energy efficiency is defined as
	\begin{align}
		\frac{\sum_{k\in\mathcal{K}}B\log_2(1+\gamma_k)}{{P}_{\text{total}}(\mathcal{A}, \{ \bm{w}_k \})},
	\end{align}
	where $B = 1$ MHz is the bandwidth of each user.
	Fig.~\ref{fig:ee_diff_k} depicts the energy efficiency versus the number of RISs with different baseline algorithms.
	The locations of RISs are randomly generated around the obstacle in Fig.~\ref{fig:x-y-model}.
	From Fig.~\ref{fig:ee_diff_k}, we observe that the energy efficiency of most baseline algorithms increases with the number of RISs.
	This is because, with more RISs, the BS can determine a more efficient RIS set, i.e., an active RIS set with lower network power consumption, leading to higher energy efficiency. Moreover, the energy efficiency of the "All active" scheme decreases when the number of RISs is 5, indicating the importance of RIS selection for increasing energy efficiency in green cellular networks.

	\subsection{User Admission Control}
	
	In this subsection, we illustrate the effectiveness of the proposed DC algorithm for the user admission control problem.
	We show the performance of the proposed DC algorithm with the following baseline schemes.
	\begin{itemize}
		\item \textbf{Smoothed $L_p$}: the BS leverages the smoothed $L_p$-minimization algorithm to maximize the number of admitted users \cite{shi2016ac}.
		\item \textbf{Proposed DC with random phase shifts}: the RIS phase-shift matrices are randomly generated. 
		Therefore, we only solve the problem (38), while dropping users using the same method as the proposed DC.
		\item \textbf{Proposed DC without RIS}: we solve problem (38) using the proposed DC algorithm without any RIS.
		\item \textbf{Exhaustive search}: we solve problems (38) and (41) using the proposed DC algorithm with all combinations of users and find the maximal active user set from the feasible solutions.
		
	\end{itemize}

	Next, we study the user admission capacities of different algorithms with respect to different QoS constraints.
	Here we set the target SINRs of all users to be equal.
	In Fig.~\ref{fig:uac}, we observe that the number of admitted users of all baselines decreases as the target SINR $\gamma_k^{\text{th}}$ increases.
	The proposed DC algorithm always admits more users than the smoothed $L_p$, the SDR, the proposed DC with random phase shifts, and the proposed DC without RIS.
	In particular, the proposed DC admits almost the same number of users compared with the exhaustive search algorithm, which shows the effectiveness of the proposed DC algorithm.
	Such performance gain of the proposed DC over the smoothed $L_p$ mainly comes from its exploitation ability for the sparse structure by matrix lifting without dropping the rank-one constraint.
	The performance gain of the proposed DC compared with the random RIS phase shifts and without RIS schemes demonstrates the effectiveness of deploying multiple RISs and optimizing the on-off status of the RISs in wireless networks.

\begin{figure}
	\centering
	\begin{minipage}[t]{\linewidth}
		\centering
		\includegraphics[scale = 0.4]{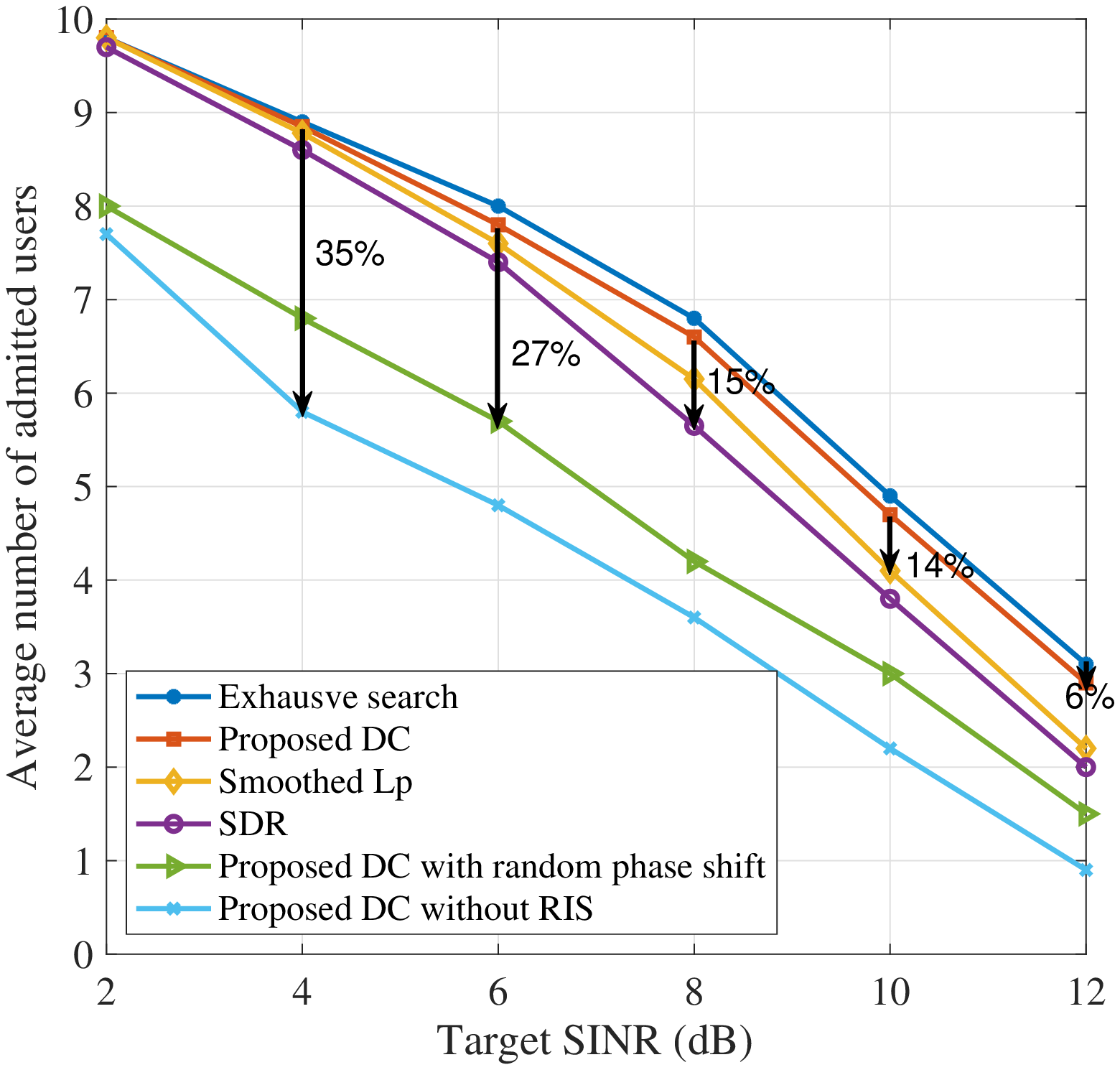}
		\caption{Average number of admitted users versus target SINR.}
		\label{fig:uac}
	\end{minipage}

	\begin{minipage}[t]{\linewidth}
		\centering
		\includegraphics[scale = 0.4]{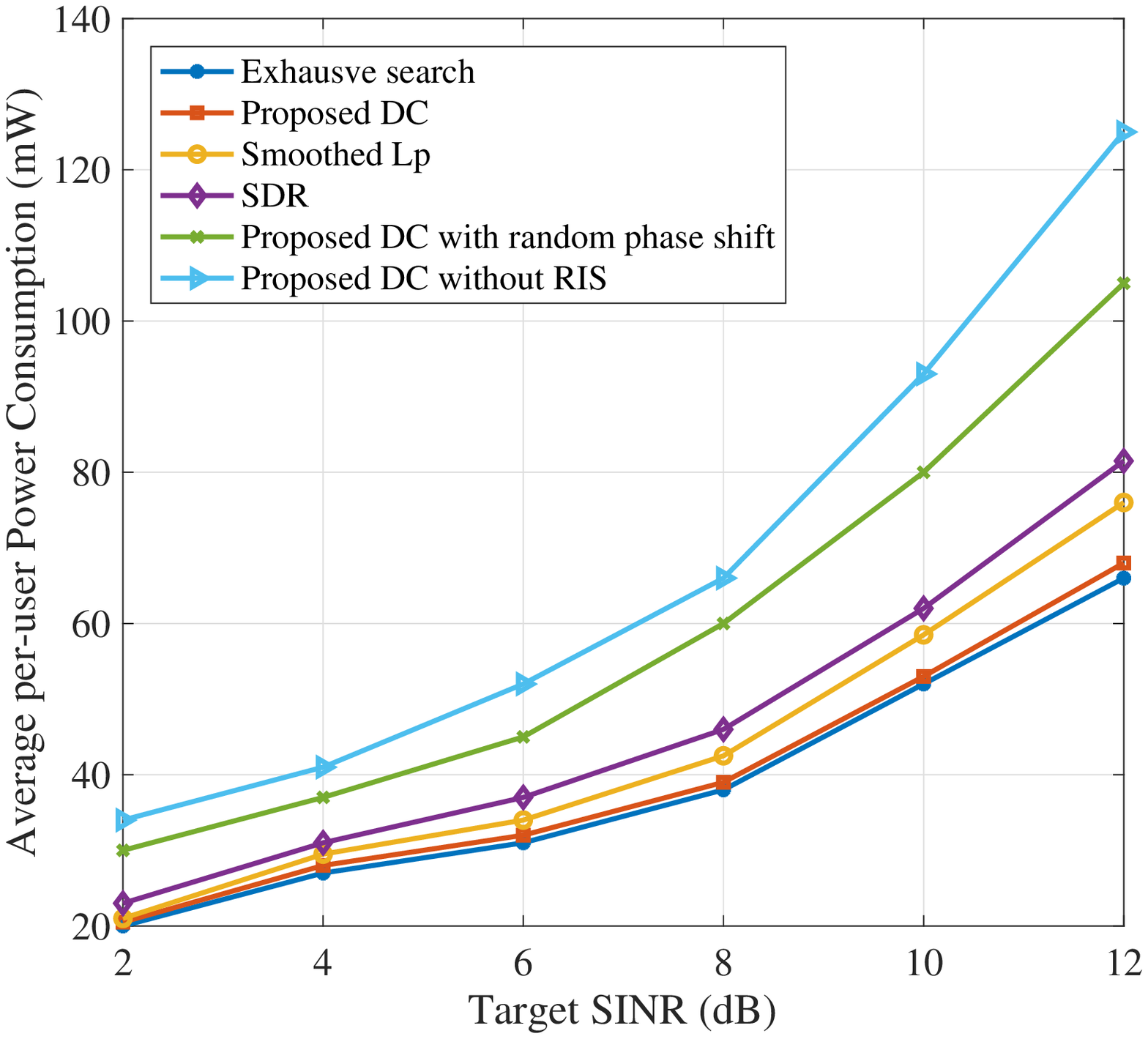}
		\caption{Average per-user power consumption versus target SINR.}
		\label{fig:pow}
	\end{minipage}

	\end{figure}

	In Fig. \ref{fig:pow}, we illustrate the per-user power consumption versus the target SINR of different algorithms.
	We observe that as the target SINR increases, the network requires higher power consumption to satisfy the QoS constraints of all users.
	The power consumption of the proposed DC algorithm is nearly the same as the exhaustive search algorithm and is much lower than other baseline schemes. 
	The obtained result demonstrates the superb performance of the proposed algorithm by exploiting the rank-one constraint.
	The proposed DC and the proposed DC with random phase shifts  scheme outperform the proposed DC without RIS, which shows the benefits of deploying RISs in terms of power control.
	Moreover, the performance loss of the proposed DC with random phase shifts compared with the proposed DC indicates the necessity to optimize RIS phase shifts for green cellular networks.
	
	\begin{figure}[t]
		\centering
		\includegraphics[scale = 0.4]{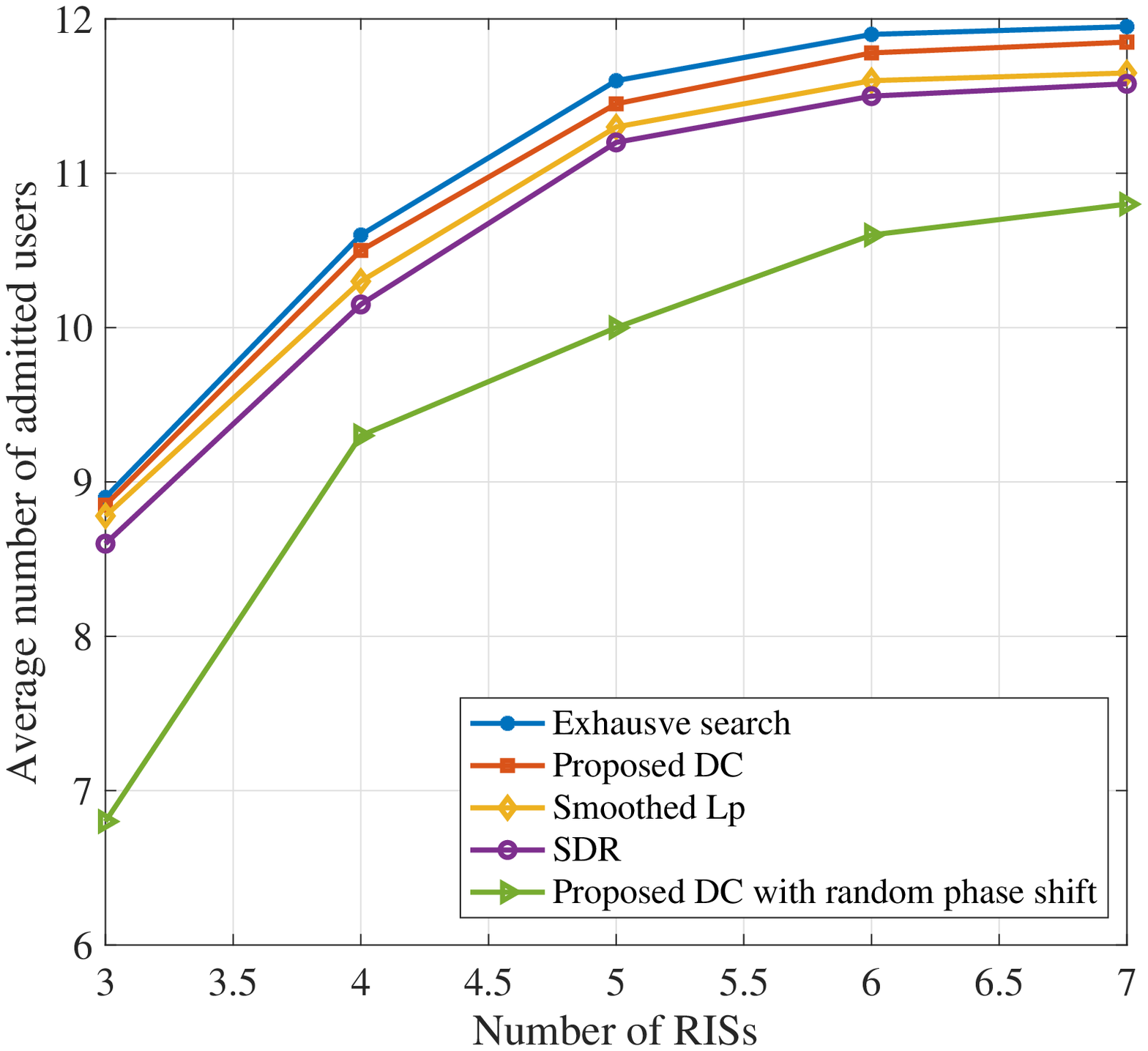}
		\caption{Average number of admitted users versus the number of RISs $L$.}
		\label{fig:npm_diff_L}
	\end{figure}

	In Fig.~\ref{fig:npm_diff_L}, we show the average number of admitted users versus the number of RISs $L$. 
	The locations of the additional RISs are randomly generated at the edge of a circle centered at $(70,40,0)$ with radius of $25$ meters. 
	Fig.~\ref{fig:npm_diff_L} demonstrates that the number of admitted users increases rapidly for a small number of RISs, however, this increase becomes slower for a larger number of RISs. 
	The reason comes from that a large number of RISs leads to high power consumption, which consequently decreases the growth of the users admission capacities. 
	Therefore, in practice, the number of RISs as well as the location should be carefully determined to avoid unnecessary implementation cost and power consumption cost.
	\section{Conclusions}
	\label{Sec_Con}
	In this paper, we investigate a downlink multi-antenna broadcast channel with multiple RISs assisting the transmission. 
		To improve the energy efficiency of the considered network, we depart from the existing works and for the first time, we jointly investigate the power control and the user admission problem.
		Specifically, we first devote to minimize the network power consumption by jointly optimizing the phase shifts of RISs, beamforming of the base stations, and the active RIS set.
		However, the infeasibility issue occurs when the QoS constraints cannot be guaranteed by all users.
		To address such infeasibility issue, we further investigate the user admission control problem. 
		A unified AO framework based on DC approach is then proposed to solve both the power minimization and the user admission control problems, followed by a low-complexity beamforming design algorithm, which acts as an alternative approach to enable the computation-sensitive wireless networks.
		Numerical results demonstrate that, the proposed DC-based AO framework significantly improves the energy efficiency. 
		Surprisingly, the proposed algorithm is capable of saving 25\% power compared to the setting with all RIS elements being active, while achieving almost the same performance as the exhaustive search scheme. In addition, for a given target SINR, the proposed DC algorithm can support more active users, resulting in a 14\% increase in user admission compared to the compared SDR algorithm.  
	\appendices

	\section{Proof of Lemma 1}\label{appd:lemma1}
		Recall problem (33), 
		\begin{subequations}\label{opt:ZF1}
			\begin{align}
				\mathop{\text{maximize}}_{\tilde{\bm{\theta}}_k}~~~ &|\tilde{\bm{\theta}}_k^H\bm{a}_{k,k}|\\
				\mathop{\text{subject to}}~~~&\tilde{\bm{\theta}}_k^H\bm{a}_{k,j}=0, \forall j\neq k,\label{equ:ZF_cons1}\\
				&|\tilde{\theta}_{k,i}|=1, \forall i\in\mathcal{I}_k.\label{equ:unit_cons_ZF1}
			\end{align}
		\end{subequations}
		Then, for constraint (\ref{opt:ZF1}b), we have
		\begin{align}
			\tilde{\theta}_{k,i}^*[\bm{a}_{k,j}]_i=-\sum_{r\in\mathcal{I}_k,r\neq i}[\bm{a}_{k,j}]_r\tilde{\theta}_{k,r}^*, \forall j\neq k.
		\end{align}
		Taking the absolute value at both sides yields
		\begin{align}\label{equ:app1_1}
			|\tilde{\theta}_{k,i}^*[\bm{a}_{k,j}]_i|=\Big|\sum_{r\in\mathcal{I}_k,r\neq i}[\bm{a}_{k,j}]_r\tilde{\theta}_{k,r}^*\Big|.
		\end{align}
		Further combining equation (\ref{equ:app1_1}) and constraint (\ref{opt:ZF1}c), we have
		\begin{align}
			|[\bm{a}_{k,j}]_i|&=\Big|\sum_{r\in\mathcal{I}_k,r\neq i}[\bm{a}_{k,j}]_r\tilde{\theta}_{k,r}^*\Big|\nonumber\\
			&\leq \sum_{r\in\mathcal{I}_k,r\neq i}\Big|[\bm{a}_{k,j}]_r\tilde{\theta}_{k,r}^*\Big|\nonumber\\
			&=\sum_{r\in\mathcal{I}_k,r\neq i}\Big|[\bm{a}_{k,j}]_r\Big|
		\end{align}
		where the inequality follows the triangle inequality. 
		Adding $|[\bm{a}_{k,j}]_i|$ on both sides, we have
		\begin{align}\label{equ:a1}
			2|[\bm{a}_{k,j}]_i|\neq \sum_{r\in\mathcal{I}_k}|[\bm{a}_{k,j}]_r|.
		\end{align}
		Since inequality (\ref{equ:a1}) should be satisfied for any $i\in\mathcal{I}_k$, conditions (\ref{equ:lemma1}) are obtained by using the $\max_{i\in\mathcal{I}_k}$ operation on both sides of (\ref{equ:a1}), which completes the proof.
		
	\section{Proof of Proposition 1}\label{appd:convergence}
	For the sequence $\{(\bm{\hat{\Theta}}^{t})\}$ generated by iteratively solving problem (\ref{t-DC_OP_The}), we denote the dual variables as $\bm{Y}_{\bm{\hat{\Theta}}}^{t}\in \partial_{\bm{\hat{\Theta}}^{t}}q$. 
	Due to the strong convexity of $q$, we have
	\begin{align}
		&q^{t+1}-q^{t}\geq \langle \Delta_t\bm{\hat{\Theta}},\bm{Y}_M^{t}\rangle  +\frac{\alpha}{2}\big(\|\Delta_t\bm{\hat{\Theta}}\|_F^2\big), \label{appd:eq1} \\
		&\langle \bm{\hat{\Theta}}^{t},\bm{Y}_{\hat{\Theta}}^{t} \rangle = q^{t}+{q^*}^{t}, \label{appd:fenchel_eq1}
	\end{align}
	where $\Delta_t\bm{\hat{\Theta}}=\bm{\hat{\Theta}}^{t+1}-\bm{\hat{\Theta}}^{t}$.
	Adding $g_1^{t+1}$ at both sides of (\ref{appd:eq1}), we have
	\begin{align}
		f_1^{t+1}\leq &g_1^{t+1}-q^{t}-\langle \Delta_t\bm{\hat{\Theta}},\bm{Y}_M^{t} \rangle -\frac{\alpha}{2}\big(\|\Delta_t\bm{\hat{\Theta}}\|_F^2\big). \label{appd:eq2}
	\end{align}
	
	For the update of primal variable $\{\bm{\hat{\Theta}}\}$ according to equation (\ref{t-iteration}), we have $\bm{Y}_{\hat{\Theta}}^{t}\in \partial_{\bm{\hat{\Theta}}^{t+1}} {g_1}$. 
	This implies that
	\begin{align}
		&g_1^{t}-g_1^{t+1}\geq \langle -\Delta_t\bm{\hat{\Theta}}, \bm{Y}_M^{t} \rangle+\frac{\alpha}{2}\big(\|\Delta_t\bm{\hat{\Theta}}\|_F^2\big), \label{appd:eq3} \\
		&\langle \bm{\hat{\Theta}}^{t+1},\bm{Y}_M^{t} \rangle = g_1^{t+1}+{g_1^*}^{t}. \label{appd:fenchel_eq2}
	\end{align}
	Similarly, by adding $-q^{t}$ at both sides of equation (\ref{appd:eq3}), we have
	\begin{align}
		f_1^{t}\geq &g_1^{t+1}-q^{t}+\langle  -\Delta_t\bm{\hat{\Theta}}, \bm{Y}_M^{t} \rangle+\frac{\alpha}{2}\big(\|\Delta_t\bm{\hat{\Theta}}\|_F^2\big). \label{appd:eq4}
	\end{align}
	Combine equation (\ref{appd:fenchel_eq1}) and equation (\ref{appd:fenchel_eq2}), we obtain
	\begin{equation}
		g_1^{t+1}-q^{t}+\langle  -\Delta_t\bm{\hat{\Theta}}, \bm{Y}_M^{t} \rangle ={f_1^*}^{t}, \label{appd:eq5}
	\end{equation}
	where $f_1^* = q^*-g_1^*$.
	By subtracting (\ref{appd:eq2}), (\ref{appd:eq4}) and (\ref{appd:eq5}), we have
	\begin{align}
		f_1^{t}&\geq {f_1^*}^{t}+\frac{\alpha}{2}\big(\|\Delta_t\bm{\hat{\Theta}}\|_F^2\big)\geq f_1^{t+1}+\alpha\big(\|\Delta_t\bm{\hat{\Theta}}\|_F^2\big).
	\end{align}

	The sequence $\{f_1^{t}\}$ is non-increasing. Since $f_1\geq 0$ always holds, we conclude that the sequence $\{f_1^{t}\}$ is strictly decreasing until convergence, i.e., $\lim_{t\rightarrow \infty} \big(\|\bm{\hat{\Theta}}^{t}-\bm{\hat{\Theta}}^{t+1}\|_F^2\big) = 0.$

	For every limit point, $f_1^{t+1}=f_1^{t}$, we have
	\begin{equation}
		\|\bm{\hat{\Theta}}^{t}-\bm{\hat{\Theta}}^{t+1}\|_F^2=0, f^{t+1}={f^*}^{t}=f^{t}.
	\end{equation}

	Then,
	\begin{align}
		{q^*}^{t}+q^{t+1}&={g}^{t}+g^{t+1} =\langle \bm{\hat{\Theta}}^{t+1},\bm{Y}_{\hat{\Theta}}^{t} \rangle, \bm{Y}_{\hat{\Theta}}^{t}\in\partial_{\bm{\hat{\Theta}}^{t+1}} q.
	\end{align}
	Therefore, $\bm{Y}_{\hat{\Theta}}^{t}\in \partial_{\bm{\hat{\Theta}}^{t+1}} g_1\cap\partial_{\bm{\hat{\Theta}}^{t+1}} q$, thus $(\{\bm{\hat{\Theta}}^{t+1}\})$ is a critical point of $f_1=g_1-q$. 
	Since
	\begin{align}
		\text{Avg}\Big(\|\bm{\hat{\Theta}}^{t}-\bm{\hat{\Theta}}^{t+1}\|_F^2\Big) &\leq\sum_{i=0}^{t}\frac{1}{\alpha(t+1)}(f_1^{[i]}-f_1^{[i+1]}) \leq\nonumber\\
		& \frac{1}{\alpha(t+1)}(f_1^{[0]}-f_1^{t+1}) \nonumber\\
		&\leq \frac{1}{\alpha(t+1)}(f_1^{[0]}-f_1^{\star})\nonumber,
	\end{align}
	we have
	\begin{align}
		\text{Avg}\Big(\|\bm{\hat{\Theta}}^{t}-\bm{\hat{\Theta}}^{t+1}\|_F^2\Big)&\leq \frac{f_1^{[0]}-f_1^{\star}}{\alpha(t+1)},
	\end{align}
	which completes the proof.
	
	\ifCLASSOPTIONcaptionsoff
	\newpage
	\fi
	
	\bibliographystyle{IEEEtran}
	\bibliography{ref}

\end{document}